\documentclass[11pt,onecolumn]{article}

\usepackage[english]{babel}
\usepackage[utf8x]{inputenc}
\usepackage[T1]{fontenc}

\usepackage[a4paper,top=2.8cm,bottom=3cm,left=2.5cm,right=2.5cm,marginparwidth=1.75cm]{geometry}

\usepackage{amsmath}
\usepackage{amsfonts}
\usepackage{amsthm}
\usepackage{subcaption}
\usepackage{graphicx}
\usepackage{algorithm} 
\usepackage{algorithmic}  
\usepackage[algo2e]{algorithm2e}
\usepackage[colorinlistoftodos]{todonotes}
\usepackage[colorlinks=true, allcolors=blue]{hyperref}
\usepackage{wrapfig}

\newcommand{\Var}{\mathrm{Var}}
\newtheorem{theorem}{Theorem}
\newtheorem{definition}{Definition}
\newtheorem{proposition}{Proposition}
\newtheorem{corollary}{Corollary}
\newtheorem{claim}{Claim}
\newtheorem{lemma}{Lemma}
\newtheorem{obs}{Observation}

\definecolor{darkred}{rgb}{1, 0.1, 0.3}
\definecolor{darkgreen}{rgb}{0.15, 0.65, 0.15}
\definecolor{darkblue}{rgb}{0.1, 0.1, 1}

\newcommand{\denselist}{\itemsep 0pt\parsep=1pt\partopsep 0pt}

\newcommand{\myleaves}		{{\mathrm{leaves}}}
\newcommand{\Dcost}			{{cost}}
\newcommand{\TC}			{{totalC}}
\newcommand{\BC}			{{baseC}}
\newcommand{\RC}			{{\rho}}
\newcommand{\myrelation}	{{relation\xspace}}
\newcommand{\myorder}[3]	{{\{{#1}, {#2} | {#3}\} }}
\newcommand{\triplecost}	{{triC}}
\newcommand{\mintricost}	{{minTriC}}
\newcommand{\newcost}		{{ratio-cost}\xspace}
\newcommand{\Newcost}		{{Ratio-cost}\xspace}

\newcommand{\expG}			{{\overline{G}}}

\newcommand{\algBisection}		{{\sf BuildPerfectTree}}
\newcommand{\validBisect}		{{\sf validBipartition}}
\newcommand{\aG}		{{\widehat{G}}}
\newcommand{\aV}			{{\widehat{V}}}
\newcommand{\aE}			{{\widehat{E}}}
\newcommand{\aT}			{{\widehat{T}}}
\newcommand{\av}			{x}
\newcommand{\nhat}		{{\hat{n}}}
\newcommand{\triTone}		{{type-1}\xspace}
\newcommand{\triTwo}		{{type-2}\xspace}
\newcommand{\triThree}		{{type-3}\xspace}

\newcommand{\edgeP}		{{\mathrm{P}}}
\newcommand{\expT}		{{\overline{T}}}
\newcommand{\LCA}			{{\mathrm{LCA}}}

\title{An Improved Cost Function for Hierarchical Cluster Trees}
\author{Dingkang Wang \and Yusu Wang}

\begin{document}
	
	\maketitle
	\setcounter{page}{0}
	
	\begin{abstract}
		Hierarchical clustering has been a popular method in various data analysis applications. 
		It partitions a data set into a hierarchical collection of clusters, and can provide  a global view of (cluster) structure behind data across different granularity levels. 
		A hierarchical clustering (HC) of a data set can be naturally represented by a tree, called a \emph{HC-tree}, where leaves correspond to input data and subtrees rooted at internal nodes correspond to clusters. 
		Many hierarchical clustering algorithms used in practice are developed in a procedure manner. 
		In \cite{Dasgupta_2016}, Dasgupta proposed to study the hierarchical clustering problem from an optimization point of view, and introduced an intuitive cost function for similarity-based hierarchical clustering with nice properties as well as natural approximation algorithms.  There since has been several followup work on better approximation algorithms, hardness analysis, and general understanding of the objective functions. 
		
		We observe that while Dasgupta's cost function is effective at differentiating a good HC-tree from a bad one for a fixed graph, the value of this cost function does not reflect how well an input similarity graph is consistent to a hierarchical structure. 
		In this paper, we present a new cost function, which is developed based on Dasgupta's cost function, to address this issue. 
		The optimal tree under the new cost function remains the same as the one under Dasgupta's cost function. However, the value of our cost function is more meaningful. For example, the optimal cost of a graph $G$ equals $1$ if and only if $G$ has a ``perfect HC-structure'' in the sense that there exists a HC-tree that is consistent with all triplets relations in $G$; and the optimal cost will be larger than $1$ otherwise. 
		The new way of formulating the cost function also leads to a polynomial time algorithm to compute the optimal cluster tree when the input graph has a perfect HC-structure, or an approximation algorithm when the input graph ``almost'' has a perfect HC-structure. 
		Finally, we provide further understanding of the new cost function by studying its behavior for random graphs sampled from an edge probability matrix. 
		
	\end{abstract}
	
	\newpage

	\section{Introduction}
	
	Clustering has been one of the most important and popular data analysis methods in the modern data era, with numerous clustering algorithms proposed in the literature \cite{Literature_Aggarwal}. 
	Theoretical studies on clustering have so far been focused mostly on the \emph{flat clustering} algorithms, e.g, \cite{Literature_Arthur,Literature_Balcan,Literature_Fernandez,Literature_Ghoshdastidar,Literature_Rohe}, which aim to partition the input data set into a set of $k$ (often pre-specified) number of groups, called clusters.  
	However, there are many scenarios where it is more desirable to perform \emph{hierarchical clustering}, which recursively partitions data into a hierarchical collection of clusters. 
	A hierarchical clustering (HC) of a data set can be naturally represented by a tree, called a \emph{HC-tree}, where leafs correspond to input data and subtrees rooted at internal nodes correspond to clusters. 
	Hierarchical clustering can provide a more thorough view of the cluster structure behind input data across all levels of granularity simultaneously, 
	and is sometimes better at revealing the complex structure behind modern data. 
	It has been broadly used in data mining, machine learning and bioinformatic applications; e.g, the studies of phylogenetics. 
	
	Most hierarchical clustering algorithms used in practice are developed in a \emph{procedure manner}: For example, the family of \emph{agglomerative methods} build a HC-tree bottom-up by starting with all data points in individual clusters, and then repeatedly merging them to form bigger clusters at coarser levels. Prominent merging strategies include single-linkage, average-linkage and complete-linkage heuristics. 
	The family of \emph{divisive methods} instead partition the data in a top-down manner, starting with a single cluster, and then recursively dividing it into smaller clusters using strategies based on spectral cut, $k$-means, $k$-center and so on. 
	Many of these algorithms work well in different practical situations, for example, average linkage algorithm is known as Unweighted Pair Group Method with Arithmetic Mean(UPGMA) algorithm \cite{Sneath_1973Book} commonly used in evolutionary biology for phylogenetic inference. However, it is in general not clear what the output HC-tree aims to optimize, and what one should expect to obtain. 
	This lack of optimization understanding of the HC-tree also makes it hard to decide which hierarchical clustering algorithm one should use given a specific type of input data. 
	
	The optimization formulation of the HC-tree was recently tackled by Dasgupta in \cite{Dasgupta_2016}. Specifically, given a similarity graph $G$ (which is a weighted graph with edge weight corresponding to the similarity between nodes), he proposed an intuitive cost function for any HC-tree, and defined an optimal HC-tree for $G$ to be one that minimizes this cost. 
	Dasgupta showed that the optimal tree under this objective function has many nice properties and is indeed desirable. Furthermore, while it is NP-hard to find the optimal tree, he showed that a simple heuristic using an $\alpha_n$-approximation of the sparsest graph cut will lead to an algorithm computing a HC-tree whose cost is an $O(\alpha_n \cdot \log n)$-approximation of the optimal cost. Given that the best approximation factor $\alpha_n$ for the sparsest cut is $O(\sqrt{\log n})$ \cite{Arora_2009Book}, this gives an $O(\log^{3/2} n)$-approximation algorithm for the optimal HC-tree as defined by Dasgupta's cost function. 
	The approximation factor has since been improved in several subsequent work \cite{Charikar_SODA2017,Cohen_SODA2018,Roy_NIPS2016}, and it has been shown independently in \cite{Charikar_SODA2017,Cohen_SODA2018} that one can obtain an $O(\sqrt{\log n})$-approximation.

	\paragraph{Our work.}
	For a fixed graph $G$, the value of Dasgupta's cost function can be used to differentiate ``better'' HC-trees (with smaller cost) from ``worse'' ones (with larger cost), and the HC-tree with the smallest cost is optimal. However, we observe that this cost function, in its current form, does not indicate whether an input graph has a strong hierarchical structure or not, or whether one graph has ``more'' hierarchical structure than another graph. 
	For example, consider a star graph $G_1$ and a path graph $G_2$, both with $n$ nodes and $n-1$ edges, and unit edge weights. 
	It turns out that the cost of the optimal HC-tree for the star $G_1$  is $\Theta(n^2)$, while that for path graph $G_2$ is $\Theta(n\log n)$. 
	However, the star, with all nodes connected to a center node, is intuitively more cluster-like than the path with $n-1$ sequential edges. (In fact, we will show later that a star, or the so-called linked star where there are two stars with their center vertices linked by an edge, see Figure \ref{fig:graph_compare} (a) in Appendix, both have what we call \emph{perfect HC-structure}.)
	Furthermore, consider a dense unit-weight graph with $\Theta(n^2)$ edges, one can show that the optimal cost is always $\Theta(n^3)$, whether the graph exhibit any cluster structure at all. 
	Hence in general, it is not meaningful to compare the optimal HC-tree costs \emph{across} different graphs. 
	
	We propose a modification of Dasgupta's cost function to address this issue and study its properties and algorithms. 
	In particular, by reformulating Dasgupta's cost function, we observe that for a fixed graph, there exists a \emph{base-cost} which reflects the minimum cost one can hope to achieve. Based on this observation, we develop a new cost $\rho_G(T)$ to evaluate  how well a HC-tree $T$ represents an input graph $G$. 
	An optimal HC-tree for $G$ is the one minimizing $\rho_G(T)$.  
	The new cost function has several interesting properties: 
	\begin{itemize}\denselist
		\item[(i)] For any graph $G$, a tree $T$ minimizes $\rho_G(T)$ for a fixed graph $G$ if and only if it minimizes Dasgupta's cost function; thus  the optimal tree under our cost function remains the same as the optimal tree under Dasgupta's cost function. Furthermore, hardness results and the existing approximation algorithm developed in \cite{Charikar_SODA2017} still apply to our cost function. 
		\item[(ii)] For any positively weighted graph $G$ with $n$ vertices,
		the optimal cost $\rho_G^* := min_{T} \rho_G(T)$ is bounded with $\rho_G^* \in [1, n-2]$ (while the optimal cost under Dasgupta's cost function could be made arbitrarily large). 
		The optimal cost 
		$\rho_G^*$ intuitively indicates how much HC-structure the graph $G$ has. 
		In particular, $\rho_G^*= 1$ \emph{if and only if} there exists a HC-tree that is consistent with all triplets relations in $G$ (see Section \ref{sec:newcostfunction} for more precise definitions), in which case, we say that this graph has a \emph{perfect hierarchical-clustering (HC) structure}. 
		\item[(iii)] The new formulation enables us to develop an $O(n^4 \log n)$-time algorithm to test whether an input graph $G$ has a perfect HC-structure (i.e, $\rho_G^*=1$) or not, as well as computing an optimal tree if $\rho_G^*=1$. 
		If an input graph $G$ is what we call the $\delta$-perturbation of a graph $G^*$ with a perfect HC-structure, 
		then in $O(n^3)$ time we can compute a HC-tree $T$ whose cost is a ($\delta^2+1$)-approximation of the optimal one.

	\item[(iv)] Finally, in Section \ref{sec:random_graphs}, we study the behavior of our cost function for a random graph $G$ generated from an edge probability matrix $\edgeP$. 
	Under mild conditions on $\edgeP$, we show that the optimal cost $\rho_G^*$ concentrates on a certain value, which interestingly, is \emph{different} from the optimal $\rho^*$ cost for the expectation-graph (i.e, the graph whose edge weights equal to entries in $\edgeP$). Furthermore, for random graphs sampled from probability matrices, the optimal cost $\rho_G^*$ will decrease if we strengthen in-cluster connections. For instance, the optimal cost of a Erd\H{o}s-R\'enyi random graph with connection probability $p$ is $\Theta(\frac{1}{p})$. In other words, the optimal cost reflects how much HC-structure a random graph has.
	\end{itemize}
	
	In general, we believe that the new formulation and results from our investigation help to reveal insights on hierarchical structure behind graph data.
	We remark that \cite{Cohen_SODA2018} proposed a concept of \emph{ground-truth input (graph)}, which, informally, is a graph consistent with an ultrametric under some monotone transformation of weights. 
	Our concept of graphs with perfect HC-structure is more general than their ground-truth graph (see Theorem \ref{thm:ground_truth_is_HC_perfect}), and for example are much more meaningful for unweighted graphs (Proposition \ref{prop:groundtruth-unweighted}).

	\paragraph{More on related work.} 
	The present study is inspired by the work of Dasgupta \cite{Dasgupta_2016}, as well as the subsequent work in \cite{Charikar_SODA2017,Cohen_SODA2018,Moseley_NIPS2017,Roy_NIPS2016}, especially  \cite{Cohen_SODA2018}. 
	As mentioned earlier, after \cite{Dasgupta_2016}, there have been several independent follow-up work to improve the approximation factor to $O(\log n)$ \cite{Roy_NIPS2016} and then to $O(\sqrt{\log n})$ via more refined analysis of the sparse-cut based algorithm of Dasgupta \cite{Charikar_SODA2017,Cohen_SODA2018}, or via SDP relaxation \cite{Charikar_SODA2017}. 
	It was also shown in \cite{Charikar_SODA2017,Roy_NIPS2016} that it is hard to approximate the optimal cost (under Dasgupta's cost function) within any constant factor, assuming the Small Set Expansion (SSE) hypothesis originally introduced in \cite{Raghavendra_2010}. 
	A dual version of Dasgupta's cost function was considered in \cite{Moseley_NIPS2017}; and an analog of Dasgupta's cost function for dissimilarity graphs (i.e, graph where edge weight represents dissimilarity) was studied in \cite{Cohen_SODA2018}. 
	In both cases,  the formulation leads to a maximization problem, and thus exhibits rather different flavor from an approximation perspective: Indeed, simple constant-factor approximation algorithms are proposed in both cases. 
	
	We remark that the investigation in \cite{Cohen_SODA2018} in fact targets a broader family of cost functions than Dasgupta's cost function (and thus ours as well),

	\section{An Improved Cost Function for HC-trees and Properties}
	\label{sec:newcostfunction}
	
	In this section, we first describe the cost function proposed by Dasgupta \cite{Dasgupta_2016} in Section \ref{subsec:Dasgupta}. We introduce our cost function in Section \ref{subsec:newcost} and present several properties of it in \ref{subsec:properties}. 
	
	\subsection{Problem setup and Dasgupta's cost function}
	\label{subsec:Dasgupta}
	
	Our input is a set of $n$ data points $V = \{v_1, \ldots, v_n\}$ as well as their pairwise similarity, represented as a $n\times n$ weight matrix $W$ with $w_{ij} = W[i][j]$ representing the similarity between points $v_i$ and $v_j$. We assume that $W$ is symmetric and each entry is non-negative. 
	Alternatively, we can assume that the input is a weighted (undirected) graph $G = (V, E, w)$, with the weight for an edge $(v_i, v_j) \in E$ being $w_{ij}$.  
	The two views are equivalent: if the input graph $G$ is not a complete graph, then in the weight matrix view, we simply set $w_{ij} = 0$ if $(v_i, v_j) \notin E$. 
	We use the two views interchangeably in this paper. Finally, we use ``unweighted graph'' $G=(V,E)$ to refer to a graph where each edge $e\in E$ has unit-weight $1$. 
	
	Given a set of data points $V = \{v_1, \ldots, v_n\}$, 
	a \emph{hierarchical clustering tree (HC-tree)} is a rooted tree $T = (V_T, E_T)$ whose leaf set equals $V$. We also say that $T$ is a HC-tree \emph{spanning (its leaf set) $V$}. 
	Given any tree node $u \in V_T$, $T[u]$ represents the subtree rooted at $u$, and $\myleaves(T[u])$ denotes the set of leaves contained in the subtree $T[u]$. 
	Given any two points $v_i, v_j \in V$, 
	$\LCA_T(i,j)$ denotes the lowest common ancestor of leafs $v_i$ and $v_j$ in $T$; the subscript $T$ is often omitted when the choice of $T$ is clear. 
    To simplify presentation, we sometimes use indices of vertices to refer to vertices, e.g, $(i, j)\in E$ means $(v_i, v_j)\in E$. 
	The following cost function to evaluate a HC-tree $T$ w.r.t. a similarity graph $G$ was introduced in \cite{Dasgupta_2016}:
	\[
	\Dcost_G(T) = \sum_{(i, j) \in E} w_{ij} | \myleaves(T[\LCA(i,j)])|. 
	\]
	
	An optimal HC-tree $T^*$ is defined as one that minimizing $\Dcost_G(T)$. 
	Intuitively, to minimize the cost, pairs of nodes with high similarity should be merged (into a single cluster) earlier. 
	It was shown in \cite{Dasgupta_2016} that the optimal tree under this cost function has several nice properties, e.g, behaving as expected for graphs such as disconnected graphs, cliques, and planted partitions. 
	In particular, if $G$ is an unweighted clique, then all trees have the same cost, and thus are optimal -- this is intuitive as no preference should be given to any specific tree shape in this case.

	\subsection{The new cost function}
	\label{subsec:newcost}
	
	\begin{wrapfigure}{r}{0.45\textwidth}
		\centering 
		\includegraphics[width=0.45\textwidth]{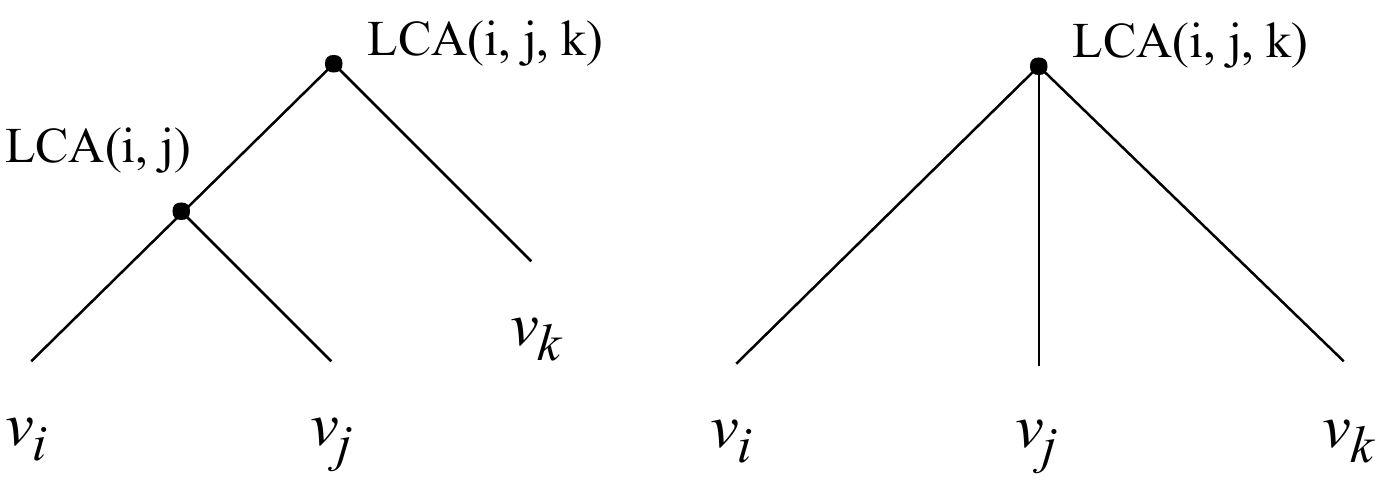}
		\caption{ $\{i, j | k\}$ (left) and $\{i | j | k\}$ (right)}
		\label{fig:merge_order}
	\end{wrapfigure}
	To introduce our new cost function, it is more convenient to take the matrix view where the weight $w_{ij}$ is defined for all pairs of nodes $v_i$ and $v_j$ (as mentioned earlier, if the input is a weighted graph $G=(V,E,w)$, then we set $w_{ij} = 0$ for $(v_i, v_j)\notin E$). 
	First, a triplet $\{i, j, k\}$ means three distinct indices $i\neq j \neq k\in [1,n]$. We say that \emph{\myrelation{} $\{i,j | k\}$ holds in $T$}, if the lowest common ancestor $\LCA(i,j)$ of $v_i$ and $v_j$ is a proper descendant of $\LCA(i, j, k)$. 
	Intuitively, subtrees containing leaves $v_i$ and $v_j$ will merge first, before they merge with a subtree containing $v_k$. 
	We say that \emph{\myrelation{}$\{i | j | k\}$ holds in $T$}, if they are merged at the same time; that is, $\LCA(i,j) = \LCA(j,k) = \LCA(i,j,k)$. See Figure \ref{fig:merge_order} for an illustration.

	\begin{definition}\label{def:total-cost}
		Given any triplet $\{i,j,k\}$ of $[1,n]$, the \emph{cost of this triplet (induced by $T$)} is $$\triplecost_{T, G}(i,j,k) = \begin{cases} w_{ik} + w_{jk} & ~~~~\mbox{if relation}~\myorder{i}{j}{k}~\mbox{holds} \\
		w_{ij}+w_{jk} & ~~~~\mbox{if relation}~\myorder{i}{k}{j}~\mbox{holds}  \\
		w_{ij}+w_{ik} & ~~~~\mbox{if relation}~\myorder{j}{k}{i}~\mbox{holds}  \\
		w_{ij}+w_{jk}+w_{ik} & ~~~~\mbox{if relation}~\{i | j | k\}~\mbox{holds}
		\end{cases}$$ 
		We omit $G$ from the above notation when its choice is clear. The \emph{total-cost of tree $T$ w.r.t. $G$} is
		$$\TC_G(T) = \sum_{i\neq j \neq k \in [1,n]} \triplecost_T(i,j,k).$$
	\end{definition}
	
	The rather simple proof of the following claim is in Appendix \ref{appendix:claim:totalCrelation}. 
	\begin{claim}\label{claim:totalCrelation}
		$\TC_G(T) = \sum_{(i, j) \in E} w_{ij} (| \myleaves(T[\LCA(i,j)])| - 2) = \Dcost_G(T) - 2\sum_{(i,j)\in E} w_{ij}. $
	\end{claim}
	
	The total-cost of any HC-tree $T$ differs from Dasgupta's cost $\Dcost_G(T)$ by a quantity depending only on the input graph $G$. Hence for a fixed graph $G$, it maintains the relative order of the costs of any two trees, implying that the optimal tree under $\TC_G$ or $\Dcost_G$ remains the same. 
	
	While the difference from Dasgupta's cost function seems to be minor, it is easy to see from this formulation that for a fixed graph $G$, there is a least-possible cost that any HC-tree will incur, which we call the base-cost. It is important to note that the following base-cost depends \emph{only} on the input graph $G$. 
	\begin{definition}
		Given a $n$-node graph $G$ associated with similarity matrix $W$, for any distinct triplet $\{i,j,k \} \subset [1,n]$, define its \emph{min-triplet cost} to be
		$$\mintricost_G(i,j,k) = \min \{ w_{ij} + w_{ik}, w_{ij} + w_{jk}, w_{ik} + w_{jk} \}. $$
		The \emph{base-cost of similarity graph $G$} is 
		$$\BC(G) = \sum_{i\neq j\neq k \in [1,n]} \mintricost_G(i,j,k).
		$$
	\end{definition}
	
	To differentiate from Dasgupta's cost function, we call our new cost function the \newcost{}.
	\begin{definition}[\Newcost{} function]
		Given a similarity graph $G$ and a HC-tree $T$, the \emph{\newcost{} of $T$ w.r.t. $G$} is defined as
		$$\RC_G(T) = \frac{\TC_G(T)}{\BC(G)}.$$
		The \emph{optimal tree} for $G$ is a tree $T^*$ such that $\RC_G(T^*) = \min_{T} \RC_G(T)$; and its \newcost{} $\RC_G(T^*)$ is called the \emph{optimal \newcost{} $\RC_G^* = \RC_G(T^*)$}. 
	\end{definition}
	
	\begin{obs}\label{obs:lowerbnd}
		(i) For any HC-tree $T$, $\TC_G(T) \ge \BC(G)$, implying that $\rho_G(T) \ge 1$.  \\
		(ii) A tree optimizing $\RC_G$ also optimizes $\TC_G$ (and thus $\Dcost_G$), as $\BC(G)$ is a constant for a fixed graph. \\
		(iii) There is always an optimal tree for $\RC_G$ that is binary. 
	\end{obs}
	Observation (i) and (ii) above follow directly the definitions of these costs and Claim \ref{claim:totalCrelation}. 
	(iii) holds as it holds for Dasgupta's cost function $\Dcost_G$ (\cite{Dasgupta_2016}, Section 2.3). Hence in the remainder of the paper, we will consider \emph{only} binary trees when talking about optimal trees. 
	
	Note that it is possible that $\BC(G) = 0$ for an non-empty graph $G$. We follow the convention that $\frac{0}{0}= 1$ while $\frac{x}{0} = +\infty$ for any positive number $x>0$. 
	We show in Appendix \ref{app:tc_bc_0} that in this case, there must exist a tree $T$ such that $\TC_G(T) =0$. Thus $\RC_G^* = 1$ for this case. 
	
	\paragraph{Intuition behind the costs.}
	Consider any  triplet $\{i, j, k\}$, and assume w.l.o.g that $w_{ij}$ is the largest among the three pairwise similarities. 
	If there exists a ``perfect'' HC-tree $T$, then it should first merge $v_i$ and $v_j$ as they are most similar, before merging them with $v_k$. That is, the \myrelation{} for this triplet in the HC-tree should be $\myorder{i}{j}{k}$; and we say that this \myrelation{} $\myorder{i}{j}{k}$ (and the tree $T$) is \emph{consistent} with (similarties of) this triplet. 
	The cost of the triplet $\triplecost_T(i,j,k)$ is designed to reward this ``perfect'' \myrelation: $\triplecost_T(i,j,k)$ is minimized, in which case $\triplecost_T(i,j,k) = \mintricost(i,j,k)$, only when $\myorder{i}{j}{k}$ holds in $T$. 
	In other words, $\mintricost(i,j,k)$ is the smallest cost possible for this triplet, and a HC-tree $T$ can achieve this cost only when the \myrelation{} of this triplet in $T$ is consistent with their similarities. 
	
	If there is a HC-tree $T$ such that for all triplets, their \myrelation{}s in $T$ are consistent with their similarities, then $\TC_G(T) = \BC(G)$, implying that $T$ must be optimal as $\rho_G^* = 1$. 
	Similarly, if $\rho_G^* = 1$ for a graph $G$, then the optimal tree $T^*$ has to be consistent with all triplets from $V$. Intuitively, this graph $G$ has a perfect HC-structure, in the sense that there is a HC-tree such that the desired merging order for all triplets are preserved in this tree. 
	\begin{definition}[Graph with perfect HC-structure]\label{def:perfectHC}
		A similarity graph $G$ has \emph{perfect HC-structure} if $\rho_G^* = 1$. Equivalently, there exists a HC-tree $T$ such that $\TC_G(T) = \BC(G)$. 
	\end{definition}
	
	\paragraph{Examples.}
	The base-cost is independent of tree $T$ and can be computed easily for a graph. If we find a tree $T$ with $\TC_G(T) = \BC(G)$, then $T$ must be optimal and $G$ has perfect HC-structure. 
	
	With this in mind, it is now easy to see that for a clique $G$ (with unit weight), for any HC-tree $T$ and any triplet $\{i,j,k\}$, $\triplecost_T(i,j,k) = 2 = \mintricost(i,j,k)$. 
	Hence $\TC_G(T) = \BC(G) = 2\cdot \binom{n}{3}$ for any HC-tree $T$, and thus the clique has perfect HC-structure. 
	A complete graph whose edge weights equal to entries of the edge probabilities of a planted partition also has a perfect HC-structure. 
	It is also easy to check that the $n$-node star graph $G_1$ (or two linked stars) has  perfect HC-structure, while for the $n$-node path $G_2$, $\RC^*_{G_2} = \Theta(\log n)$. Intuitively, a path does not process much hierarchical structure. See Appendix \ref{appendix:examples:costs} for details. 
	We also refer the readers to see more results and discussions on Erd\"{o}s R\'{e}nyi random graph and planted bipartiton (also called planted bisection) random graphs in Section \ref{sec:random_graphs}. In particular, as we describe in the discussion after Corollary \ref{coro:2}, the \newcost{} function exhibits an interesting, yet natural, behavior as the in-cluster and between-cluster probabilities for a planted bisection model change. 
	
	\vspace*{0.08in}\noindent{\emph{\underline{Remark:}}}
	Note that in general, the value of Dasgupta's cost $\Dcost_G(T)$ is affected by the (edge) density of a graph. One may think that we could normalize $\Dcost_G(T)$ by the number of edges in the graph (or by total edge weights). However, the examples of the star and path show that this strategy itself is not sufficient to reveal the cluster structure behind an input graph, as those two graphs have equal number of edges. 
	
	\subsection{Some properties}
	\label{subsec:properties}
	
	For an unweighted graph, optimal cost under Dasgupta's cost function is bounded by $O(n^3)$. But this cost can be made arbitrarily large for a weighted graph. 
	For our \newcost{} function, it turns out that the optimal cost is always bounded for both unweighted and weighted graphs. Proof of the following result can be found in Appendix \ref{appendix:thm:rhoranges}. 
	For an unweighted graph, we can in fact obtain an upper bound that is asymptotically the same as the one in (ii) below using a much simpler argument (than the one in our current proof). 
	However, the stated upper bound ($\frac{n^2-2n}{2m-n}$) is tight in the sense that it equals `1' for a clique, matching the fact that $\rho_G^* = 1$ for a clique.  
	\begin{theorem} \label{thm:rhoranges}
		(i) Given a similarity graph $G= (V, E, w)$ with $w$ being symmetric, having non-negative entries, we have that 
		$\rho_G^* \in [1, n-2]$ where $n = |V| \ge 3$. \\
		(ii) For a connected unweighted graph $G = (V, E)$ with $n = |V|$ and $m = |E|$, we have that $\rho_G^* \in [1, \frac{n^2-2n}{2m-n}]$.
	\end{theorem}
	We now show that the bound in (i) above is asymptotically tight. To prove that, we will use a certain family of edge expanders. 
	\begin{definition}	
		Given $\alpha > 0$ and integer $k$, a graph $G(V, E)$ is an
		\emph{$(\alpha, k)$-edge expander} if for every $S \subseteq V$ with $|S| \leq k$, the set of crossing edges from $S$ to $V \backslash S$,
		denoted as $E(S, \bar{S})$, has cardinality at least $\alpha|S|$.
	\end{definition}
	A (undirected) graph is \emph{$d$-regular} if all nodes have degree $d$. 
	\begin{theorem} \cite{Arora_2009Book} \label{thm:existence_edge_expanders}
		For any natural number $d \geq 4$, and sufficiently large $n$, there exist $d$-regular graphs
		on $n$ vertices which are also $(\frac{d}{10}, \frac{n}{2})$-edge expanders.
	\end{theorem}
	
	\begin{theorem}\label{thm:tightbnds}
		A $d$-regular $(\frac{d}{10}, \frac{n}{2})$-edge expander graph $G$ satisfies that $\rho_G^* = \Theta(n)$. 
	\end{theorem}
	\begin{proof}
		The existence of a $d$-regular $(\frac{d}{10}, \frac{n}{2})$ edge expander is guaranteed by Theorem \ref{thm:existence_edge_expanders}.  
		Given a cut $(S, V\setminus S)$, the sparsity of it is $\phi(S) = \frac{|E(S, V\setminus S)|}{|S|\cdot |V\setminus S|}$. 
		Consider the sparsest-cut $(S^*, V\setminus S^*)$ for $G$ with minimum sparsity. Assume w.l.o.g that $|S^*| \le n/2$. Its sparsity satisfies: 
		\begin{align}
		\phi^* &= \frac{|E(S^*, V\setminus S^*)|}{|S^*| \cdot |V\setminus S^*|} \ge \frac{\frac{d}{10}\cdot |S^*|}{|S^*| \cdot |V\setminus S^*|} = \frac{d}{10 |V\setminus S^*|} \ge \frac{d}{10n}. \label{eqn:sparsecut}
		\end{align}
		The first inequality above holds as $G$ is a $(\frac{d}{10}, \frac{n}{2})$-edge expander. 
		
		On the other hand, consider an arbitrary bisection cut $(\widehat{S}, V\setminus \widehat{S})$ with $|\widehat{S}| = n/2$. Its sparsity satisfies: 
		\begin{align}
		\phi(\widehat{S}) &= \frac{|E(\widehat{S}, V\setminus \widehat{S})|}{|\widehat{S}| \cdot |V\setminus \widehat{S}|} = \frac{E(\widehat{S}, V\setminus \widehat{S})}{\frac{n}{2}\cdot \frac{n}{2}} \le \frac{d\cdot \frac{n}{2}}{n^2/4} = \frac{2d}{n}. \label{eqn:bisection}
		\end{align}
		Combining (Eqn. \ref{eqn:sparsecut}) and (Eqn. \ref{eqn:bisection}), we have that the bisection cut $(\widehat{S}, V\setminus \widehat{S})$ is a $20$-approximation for the sparest cut. 
		On the other hand, 
		it is shown in \cite{Cohen_SODA2018} that a recursive $\beta$-approximation sparest-cut algorithm will lead to a HC-tree $T$ whose cost $\Dcost_G(T)$ is a $c\cdot \beta$-approximation for $\Dcost_G(T^*)$ where $T^*$ is the optimal tree. 
		Now consider the HC-tree $T$ resulted from first bisecting the input graph via $(\widehat{S}, V\setminus \widehat{S})$, then identifying sparest-cut for each subtree recursively. 
		By \cite{Cohen_SODA2018}, $\Dcost_G(T) \le c\cdot 20 \cdot \Dcost_G(T^*)$ for some constant $c$. 
		Furthermore, $\Dcost_G(T)$ is at least the cost induced by those edges split at the top level; that is, 
		\begin{align}
		\Dcost_G(T) &= \sum_{(i,j)\in E} |\myleaves(T[\LCA(i, j)])| \ge \sum_{(i,j) \in E(\widehat{S}, V\setminus \widehat{S})} |\myleaves(T[\LCA(i, j)])| \nonumber \\
		&= \sum_{(i,j) \in E(\widehat{S}, V\setminus \widehat{S})} n = n \cdot |E(\widehat{S}, V\setminus \widehat{S})| \ge n \cdot \frac{nd}{20} = \frac{dn^2}{20}, 
		\end{align}
		where the last inequality holds as $G$ is a $(\frac{d}{10}, \frac{n}{2})$-edge expander. Hence 
		$$
		\Dcost_G(T^*) \ge \frac{1}{20c} \Dcost_G(T) \ge \frac{dn^2}{400c}, 
		$$
		with $c$ being a constant. 
		By Claim \ref{claim:totalCrelation}, we have that 
		\begin{align}
		\TC_G(T^*) = \Dcost_G(T^*) - 2|E| \ge \frac{dn^2}{400c} - dn = \Omega(dn^2). \label{eqn:totalexpander}
		\end{align}
		
		Now call a triplet \emph{a wedge} (resp. {\it a triangle} if there are exactly two (resp. exactly three) edges among the three nodes. Only wedges and triangles contribute to $\BC(G)$. Thus: 
		\begin{align*}
		\BC(G) &= \# wedges + 2 \# triangles \le \sum_{i\in [1,n]} \binom{d}{2} = O(d^2 n)  \\
		\Rightarrow~~~~~ \rho_G^* &= \frac{\TC_G(T^*)}{\BC(G)} = \Omega(\frac{dn^2}{d^2 n}) = \Omega(\frac{n}{d}). 
		\end{align*}
		Combining the above with Theorem \ref{thm:rhoranges}, the claim then follows. 
	\end{proof}

	\paragraph{Relation to the ``ground-truth input'' graphs of \cite{Cohen_SODA2018}.}
	Cohen-Addad et al. introduced what they call the ``ground-truth input'' graphs to describe inputs that admit a ``natural'' ground-truth cluster tree \cite{Cohen_SODA2018}. 
	A brief review of this concept is given in Appendix \ref{appendix:groundtruthinput}. 
	Interestingly, we show that a ground-truth input graph always has a perfect HC-structure. However, the converse is not necessarily true and the family of graphs with perfect HC-structure is broader while still being natural. 
	Proof of the following theorem can be found in Appendix \ref{appendix:groundtruthinput}. 
	
	\begin{theorem} \label{thm:ground_truth_is_HC_perfect}
		Given a similarity graph $G = (V, E, w)$, if $G$ is a ground-truth input graph of \cite{Cohen_SODA2018}, then $G$ must have perfect HC-structure. However, the converse is not true. 
	\end{theorem}
	
	\begin{wrapfigure}{r}{0.45\textwidth}
		\vspace*{-0.1in}
		\centering 
		\includegraphics[width=0.45\textwidth]{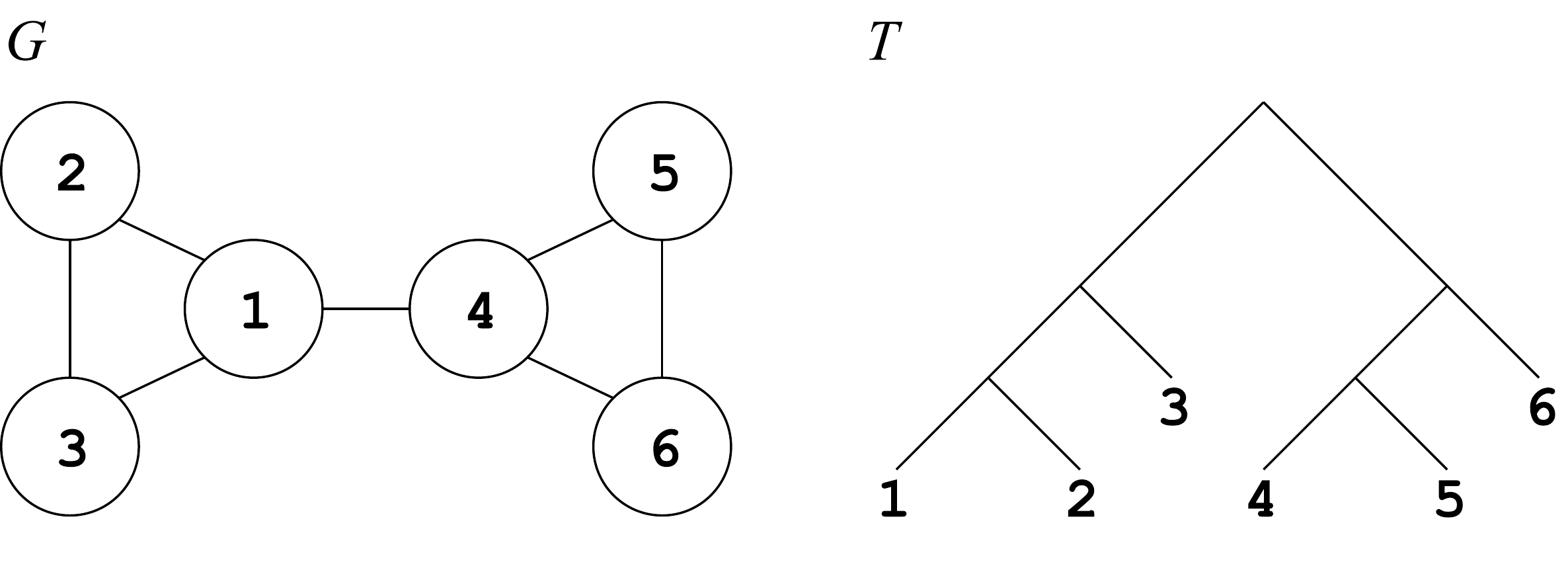}
		\vspace*{-0.2in}
		\caption{A unweighted graph $G$ with a perfect HC-structure; an optimal HC-tree $T$ is shown on the Right.
			This graph is not a ground-truth input of \cite{Cohen_SODA2018}.} 
		\label{fig:linkage_not_working}
	\end{wrapfigure}
	\noindent Intuitively, a graph $G$ has a perfect HC-structure if there exists a tree $T$ such that for all triplets, the most similar pair (with the largest similarity) will be merged first in $T$. 
	Such a triplet-order constraint is much weaker than the requirement of the ground-truth graph of \cite{Cohen_SODA2018} (which intuitively is generated by an ultrametric). An example is given in Figure \ref{fig:linkage_not_working}. 
	In fact, the following proposition 
	shows that the concept of ground-truth graph is rather stringent for graphs with unit edge weights (i.e, unweighted graphs). In particular, a connected unweighted graph is a ground-truth graph of \cite{Cohen_SODA2018} if and only if it is the complete graph. 
	In contrast, unweighted graphs with perfect HC-structure represent a  much broader family of graphs. 
	The proof of this proposition is in Appendix \ref{appendix:groundtruthinput}. 
	
	\begin{proposition}\label{prop:groundtruth-unweighted}
		Let $G = (V, E)$ be an unweighted graph (i.e, $w(u,v) = 1$ if $(u,v)\in E$ and $0$ otherwise).
		$G$ is a ground-truth graph \emph{if and only if} each connected component of $G$ is a clique. 
	\end{proposition}
	
	\section{Algorithms}
	\label{sec:algorithms}
	
	\subsection{Hardness and approximation algorithms}
	By Claim \ref{claim:totalCrelation}, $\TC_G(T)$ equals $\Dcost_G(T)$ minus a constant (depending only on $G$). 
	Thus the hardness results for optimizing $\Dcost_G(T)$ also holds for optimizing $\rho_G(T) = \frac{\TC_G(T)}{\BC(G)}$. 
	Hence the following theorem follows from results of \cite{Charikar_SODA2017} and \cite{Dasgupta_2016}. The simple proof is in Appendix \ref{appendix:thm:hardness}. 
	\begin{theorem}\label{thm:hardness}
		(1) It is NP-hard to compute $\rho_G^*$, even when $G$ is an unweighted graph (i.e, edges have unit weights). 
		(2) Furthermore, under Small Set Expansion hypothesis, it is NP-hard to approximate $\rho_G^*$ within any constant factor. 
	\end{theorem}

	\noindent We remark that while the hardness results for optimizing $\Dcost_G(T)$ translate into hardness results for $\rho_G^*$, it is not immediately evident that an approximation algorithm translates too, as $\TC_G(T)$ differs from $\Dcost_G(T)$ by a positive quantity. 
	Nevertheless, it turns out that the $O(\sqrt{\log n})$-approximation algorithm of \cite{Charikar_SODA2017} for $\Dcost_G(T^*)$ also approximates $\rho_G^*$ within the same asymptotic approximation factor. See appendix \ref{app:sdp_app} for details.

	\subsection{Algorithms for graphs with perfect or near-perfect HC-structure} \label{sec:perfect_near_perfect}
	
	While in general, it remains open how to approximate $\rho_G^*$ (as well as $\Dcost_G(T^*)$) to a factor better than $\sqrt{\log n}$, we now show that we can check whether a graph has perfect HC-structure or not, and compute an optimal tree if it has, in polynomial time. We also provide a polynomial-time approximation algorithm for graphs with near-perfect HC-structures (to be defined later). 
	
		\noindent{\emph{\underline{Remark:}}} 
		One may wonder whether a simple agglomerative (bottom-up) hierarchical clustering algorithm, such as the single-linkage, average-linkage, or complete-linkage clustering algorithm, could have recovered the perfect HC-structure. For the ground-truth input graphs introduced in \cite{Cohen_SODA2018}, it is known that it can be recovered via average-linkage clustering algorithms. 
	However, as the example in Figure \ref{fig:linkage_not_working} shows, this in general is not true for recovering graphs with perfect HC-structures, as depending on the strategy, a bottom-up approach may very well merge nodes $1$ and $4$ first, leading to a non-optimal HC-tree. 
	
	Intuitively, it is necessary to have a \emph{top-down} approach to recover the perfect HC-structure. 
	The high level framework of our recursive algorithm \algBisection($\aG$) is given below and output a HC-tree $\aT$ spans (i.e, with its leaf set being) a subset of vertices from $\aV = V(\aG)$. 
	$\aG_A$ (resp. $\aG_B$) in the algorithm denotes the subgraph of $\aG$ spanned by vertices in $A \subseteq \aV$ (resp. in $B \subseteq \aV$). 
	We will prove that  the output tree $\aT$ spans \emph{all} vertices $V(\aG)$, if and only if $\aG$ has a perfect HC-structure (in which case $\aT$ will also be an optimal tree). 
	
	\begin{description}\denselist
		\item[\algBisection($\aG$)] $/*$ \underline{Input}: graph $\aG = (\aV, \aE)$. \underline{Output}: a binary HC-tree $\aT$ $*/$
		\begin{itemize}\denselist
			\item[] Set $(A,B)$= \validBisect($\aG$); ~~{\sf If}($A = \emptyset$ or $B = \emptyset$) {\sf Return}($\emptyset$) 
			\item[] Set $T_A$ = \algBisection($\aG_A$); ~~ $T_B$=\algBisection($\aG_B$)
			\item[]	Build tree $\aT$ with $T_A$ and $T_B$ being the two subtrees of its root. {\sf Return}($\aT$)
		\end{itemize}
	\end{description}

	\noindent We say that $(A,B)$ is a \emph{partial bi-partition of $\aV$} if $A \cap B = \emptyset$ and $A\cup B \subseteq \aV$; and $(A, B)$ is a \emph{bi-partition of $\aV$ (or $\aG$)} if  $A\cap B = \emptyset$, $A, B \neq \emptyset$, and $A\cup B = \aV$. 
	Let $\aV = \{\av_1, \ldots, \av_\nhat\}$. 
	\begin{definition}[Triplet types]
		A triplet $\{\av_i, \av_j, \av_k\}$ with edge weights $w_{ij}, w_{ik}$ and $w_{jk}$ is \\
		{\sf \triTone{}:} if the largest weight, say $w_{ij}$, is strictly larger than the other two; i.e, $w_{ij} > w_{ik}, w_{jk}$; \\
		{\sf \triTwo{}:} if exact two weights, say $w_{ij}$ and $w_{ik}$, are the largest; i.e, $w_{ij}=w_{ik} > w_{jk}$; \\
		{\sf \triThree{}:} otherwise, which means all three weights are equal; i.e,  $w_{ij} = w_{ik}= w_{jk}$. 
	\end{definition}
	\begin{definition}[Valid partition]\label{def:validpartition}
		A partition $\mathcal{S} = (S_1, \ldots, S_m)$, $m>1$, of $\aV$ (i.e, $\cup S_i = \aV$, $S_i \neq \emptyset$, and $S_i \cap S_j = \emptyset$) is \emph{valid w.r.t. $\aG$} if 
		(i) for any \triTone{} triplet $\{\av_i, \av_j, \av_k\}$ with $w_{ij} > \max\{w_{ik}, w_{jk}\}$, either all three vertices belong to the same set from $\mathcal{S}$; or $\av_i$ and $\av_j$ are in one set from $\mathcal{S}$, while $\av_k$ is in another set; and (ii) for any \triTwo{} triplet $\{\av_i, \av_j, \av_k\}$ with $\av_{ij} = \av_{ik} > \av_{jk}$, it cannot be that $\av_j$ and $\av_k$ are from the same set of $\mathcal{S}$, while $\av_i$ is in another one. 
		
		If this partition is a bi-partition, then it is also called a \emph{valid bi-partition}. 
		
		In what follows, we also refer to each set $S_i$ in the partition $\mathcal{S}$ as a \emph{cluster}. 
	\end{definition}
	
	The goal of procedure \validBisect($\aG$) is to compute a valid bi-partition if it exists. Otherwise, it returns ``{\sf fail}" (more precisely, it returns ($A = \emptyset, B=\emptyset$) ). 
	On the high level, it has two steps. It turns out (Step-1) follows from existing literature on the so-called \emph{rooted triplet consistency} problem. The main technical challenge is to develop (Step-2). 
	\begin{description}\denselist
		\item[Procedure \validBisect($\aG$)] ~
		\item[] {\sf {(Step-1)}:} Compute a certain valid partition $\mathcal{C} = \{C_1, \ldots, C_m \}$ of $\aV$, if possible. 
		\item[] {\sf {(Step-2)}:} Compute a valid bi-partition from this valid partition $\mathcal{C}$. 
	\end{description}
	
	\paragraph{(Step 1): compute a valid partition.}
	It turns out that the partition procedure within the algorithm {\sf BUILD} of \cite{Aho_1981} will compute a specific partition $\mathcal{C} = \{C_1, \ldots, C_m \}$ with nice properties. The following result can be derived from Lemma 1, Theorem 2, and proof of Theorem 4 of \cite{Aho_1981}.
	\begin{proposition}[\cite{Aho_1981}]\label{prop:goodcluster}
		(1) We can check whether there is a valid partition for $\aG$ in $O(\nhat^3 \log \nhat)$ time.
		Furthermore, if one exists, then within the same time bound we can compute a valid partition $\mathcal{C} = \{C_1, \ldots, C_m \}$ which is \emph{minimal} in the following sense: 
		Given any other valid partition $\mathcal{C}' = \{C'_1, \ldots, C'_t \}$, $\mathcal{C}$ is a refinement of $\mathcal{C'}$ (i.e, all points from the same cluster $C_i \in \mathcal{C}$ are contained within some $C'_j \in \mathcal{C}$). 
		
		(2) This implies that, let $(A, B)$ be any valid bi-partition for $\aG$, and $\mathcal{C}$ the partition as computed above. Then for any $x\in C_i$, if $x \in A$ (resp. $x\in B$), then $C_i \subseteq A$ (resp. $C_i \subseteq B$). 
	\end{proposition}
	
	\paragraph{(Step 2): compute a valid bi-partition from $\mathcal{C}$ if possible.}
	Suppose (Step 1) succeeds and let $\mathcal{C} = \{C_1, \ldots, C_m\}$, $m>1$, be the minimum valid partition computed. 
	It is not clear whether a \emph{valid bi-partition} exists (and how to find it) even though a \emph{valid partition} exists. 
	Our algorithm computes a valid bi-partition depending on whether a \emph{claw-configuration} exists or not.

	\begin{definition}[A claw]
		Four points $\{\av_i \mid \av_j, \av_k, \av_\ell \}$ form \emph{a claw w.r.t. $\mathcal{C}$} if  (i) each point is from a different cluster in $\mathcal{C}$; and (ii) 
		$w_{ij} = w_{ik} = w_{i\ell} > \max \{ w_{jk}, w_{j\ell}, w_{k\ell}\}$. See Figure \ref{fig:claw} (a). 
	\end{definition}
	
	\begin{wrapfigure}{r}{0.45\textwidth}
		\vspace*{-0.1in}
		\centering
		\begin{tabular}{cc}
			\includegraphics[height=2.5cm]{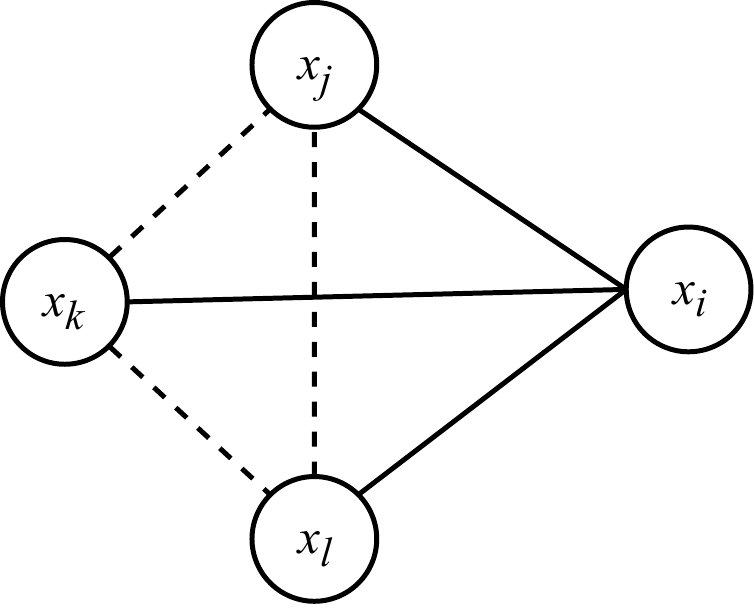}&  \includegraphics[height=3cm]{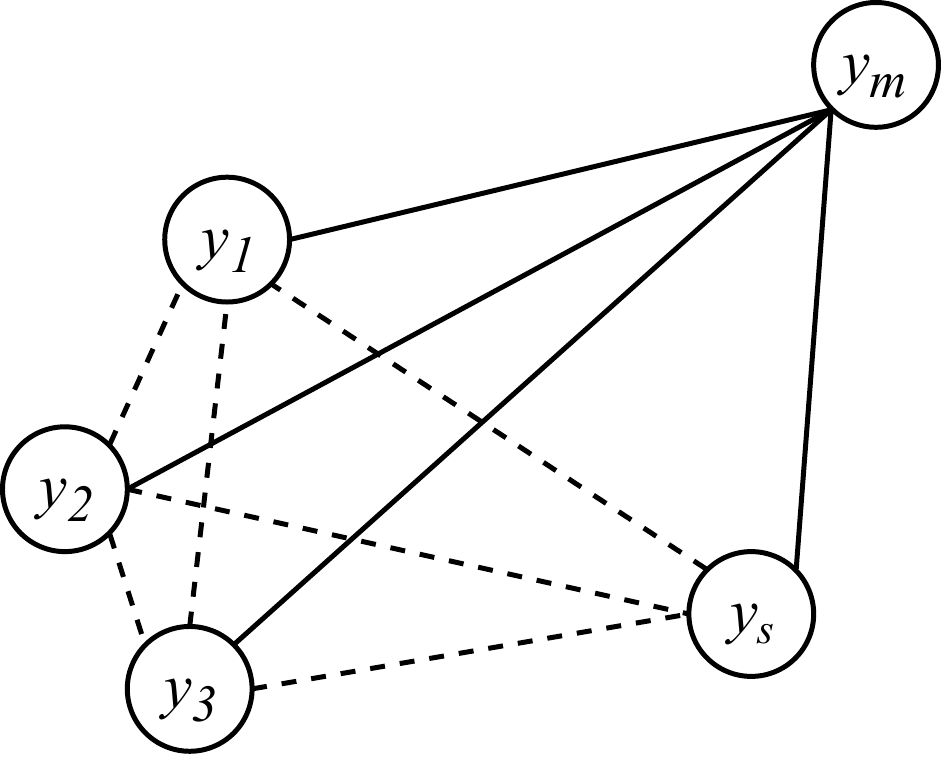}\\
			(a) & (b)
		\end{tabular}
		\vspace*{-0.15in}
		\caption{(a) A claw $\{\av_i | \av_j, \av_k, \av_\ell\}$. (b) Clique $\mathrm{C}$ formed by light (dashed) edges.
			\label{fig:claw}}
	\end{wrapfigure}
	
	\vspace*{0.05in}\noindent \emph{\underline{(Case-1)}: Suppose there is a claw w.r.t. $\mathcal{C}$.} 
	Fix any claw, and assume w.l.o.g that it is $\{y_m \mid y_1, y_2, y_3 \}$ such that $y_i \in C_i$ for $i = 1, 2, 3, $ or $m$. 
	We choose an arbitrary but fixed representative vertex $y_i \in C_i$ for each cluster $C_i$ with $i \in [4, m-1]$. 
	Compute the subgraph $G' = (V', E')$ of $\aG$ spanned by vertices $V' = \{y_1, \ldots, y_m \}$. 
	(Recall that we can view $G'$ as a complete graph where $(y_i, y_j)$ has weight $0$ if it is not an edge in $E'$.) 
	Set $\mathsf{w} = w(y_1, y_m) = w(y_2, y_m) = w(y_3, y_m)$. We say that an edge $e\in E'$ is \emph{light} if its weight is strictly less than $\mathsf{w}$; otherwise, it is \emph{heavy}. 
	Easy to see that by definition of claw, edges $y_1y_2, y_1y_3,$ and $y_2y_3$ are all light. 
	Now, consider the subgraph $G''$ of $G'$ spanned by only light-edges, and w.l.o.g. let $\mathsf{C} = \{y_1, \ldots, y_s \}$ be the maximum clique $\mathsf{C}$ in $G''$ contains $y_1, y_2$ and $y_3$. See Figure \ref{fig:claw} (b) for this clique, where solid edges are heavy, while dashed ones are light. (It turns out that this maximum clique can be computed efficiently, as we show in Appendix \ref{appendix:thm:timecomplexity}.)

	We then set up $s$ potential bi-partitions $\Pi_i = (C_i, ~\cup_{\ell \in [1,m], \ell \neq i} C_\ell)$, for each $i\in [1, s]$. 
	We check the validity of each such $\Pi_i$. 
	If none of them is valid, then procedure \validBisect($\aG$) returns `{\sf fail}'. 
	Otherwise, it turns the valid one.  
	The correctness of this step is guaranteed by the Lemma \ref{lem:withclaw}, whose proof can be found in Appendix \ref{appendix:lem:withclaw}. 
	
	\begin{lemma}\label{lem:withclaw}
		Suppose there is a claw w.r.t. $\mathcal{C}$. Let $\Pi_i$'s, $i\in [1,s]$ be described above. There exists a valid bi-partition for $\aG$ if and only if one of the $\Pi_i$, $i \in [1, s]$, is valid. 
	\end{lemma}

	\noindent \emph{\underline{(Case-2)}: Suppose there is no claw w.r.t. $\mathcal{C}$.} 
	The case when there is no claw is slightly more complicated. We present the lemma below with a constructive proof in Appendix \ref{appendix:lem:noclaw}. 
	
	\begin{lemma}\label{lem:noclaw}
		If there is no claw w.r.t. $\mathcal{C}$, then we can check whether there is a valid bi-partition (and compute one if it exists) in $O(\hat{n}^3)$ time, where $\hat{n}=|\aV|$. 
	\end{lemma}
	
	This finishes the description of procedure \validBisect($\aG$). 
	Putting everything together, we conclude with the following theorem, with proof in Appendix \ref{appendix:thm:perfectHC}. 
	We note that it is easy to obtain a time complexity of $O(n^5)$. However, we show in Appendix \ref{appendix:thm:perfectHC} how to modify our algorithm as well as to provide a much more refined analysis to improve the time to $O(n^4\log n)$. 
	
	\begin{theorem}\label{thm:perfectHC}
		Given a similarity graph $G=(V,E)$ with $n$ vertices, algorithm \algBisection($G$) can be implemented to run in $O(n^4 \log n)$ time. It returns a tree spanning \emph{all vertices in $V$} if and only if $G$ has a perfect HC-structure, in which case this spanning tree is an optimal HC-tree. 
	
		Hence we can check whether it has a perfect HC-structure, as well as compute an optimal HC-tree if it has, in $O(n^4 \log n)$ time. 
		
	\end{theorem}
	
	\paragraph{Graphs with almost perfect HC-structure.} 
	In practice, a graph with perfect HC-structure could be corrupted with noise. 
	We introduce a concept of graphs with an almost-perfect HC-structure, and present a polynomial time algorithm to approximate the optimal cost.  
	\begin{definition}[$\delta$-perfect HC-structure]
		A graph $G = (V, E, w)$ has \emph{$\delta$-perfect HC-structure}, $\delta \ge 1$, if there exists weights $w^*: E \to \mathbb{R}$ such that (i) the graph $G^* = (V, E, w^*)$ has perfect HC-structure; and (ii) for any $e = (u,v) \in E$, we have $\frac{1}{\delta} w(e) \leq w^*(e) \leq \delta \cdot w(e)$. 
		In this case, we also say that $G=(V,E,w)$ is a \emph{$\delta$-perturbation of graph $G^* =(V, E, w^*)$}. 
		
		Note that a graph with $1$-perfect HC-structure is simply a graph with perfect HC-structure. 
	\end{definition}
	 
	The proof of the following theorem is in Appendix \ref{appendix:thm:deltaperfect}. 
	Note that when $\delta =1$ (i.e, the graph $G$ has a perfect HC-structure), this theorem also gives rise to a $2$-approximation algorithm for the optimal \newcost{} in $O(n^3)$. In contrast, an exact algorithm to compute the optimal tree in this case takes $O(n^4\log n)$ time as shown in Theorem \ref{thm:perfectHC}. 
	\begin{theorem}\label{thm:deltaperfect}
		Suppose $G=(V,E,w)$ is a $\delta$-perturbation of a graph $G^* = (V, E, w^*)$ with perfect HC-structure. 
		Then we have (i) $\RC^*_{G} \leq \delta^2$; 
		and (ii) we can compute a HC-tree $T$ s.t. $\RC_{G}^{}(T) \leq (1+\delta^2) \cdot \RC_{G}^*$ (i.e, we can ($1+\delta^2$)-approximate $\RC_G^*$) in $O(n^3)$
		time. 
		
	\end{theorem}

	\section{\Newcost{} function for Random Graphs}
	\label{sec:random_graphs}

	\begin{definition}
		Given a $n \times n$ symmetric matrix $\edgeP$ with each entry $\edgeP_{ij} = \edgeP[i][j] \in [0, 1]$, $G = (V=\{v_1,\ldots,v_n\}, E)$ is a \emph{random graph generated from $\edgeP$} if there is an edge $(v_i, v_j)$ with probability $\edgeP_{ij}$. Each edge in $G$ has unit weight.
		
		The expectation-graph $\expG = (V, \overline{E}, w)$ refers to the weighted graph where the edge $(v_i, v_j)$ has weight $w_{ij}=\edgeP_{ij}$. 
	\end{definition}
	
	In all statements, ``with high probability (w.h.p)'' means ``with probability larger than $1-n^{-\varepsilon}$ for some constant $\varepsilon > 0$". 
	The main result is as follows. The proof is in Appendix \ref{appendix:thm:randomgraph}. 
	\begin{theorem}\label{thm:randomgraph}
		Given an $n\times n$ edge probability matrix $\edgeP$, assume each entry $\edgeP_{ij} = \omega(\sqrt{\frac{\log n}{n}})$, for any $i, j\in [1, n]$. 
		Given a random graph $G=(V,E)$ sampled from $\edgeP$, let $T^*$ denote the optimal HC-tree for $G$ (w.r.t. \newcost{}), and $\expT^*$ an optimal HC-tree for the expectation-graph $\expG$. Then we have that w.h.p, 
		$$\RC^*_G = \RC_G(T^*) = (1 + o(1)) \frac{\TC_{\expG}(\expT^*)}{\mathbb{E}[\BC(G)]}, $$
		where $\mathbb{E}[\BC(G)]$ is the expected base-cost for the random graph $G$.
	\end{theorem}

	\noindent Hence the otpimal cost of a randomly generated graph concentrates around a quantity. 
	However, 
	this quantity $\frac{\TC_{\expG}(\expT^*)}{\mathbb{E}[\BC(G)]}$ is \emph{different} from the optimal \newcost{} for the expectation graph $\expG$, which would be $\RC^*_{\expG} =  \frac{\TC_{\expG}(\expT^*)}{\BC(\expG)}$. 
	Rather, $\mathbb{E}[\BC(G)]$ is \emph{sensitive} to the specific random graph $G$ sampled. 
	We give two examples below. 
	A Erd\H{o}s-R\'enyi random graph $G = (n, p)$ is generated by the edge probability matrix $\edgeP$ with $\edgeP_{ij} = p$ for all $i, j\in [1,n]$. 
	The probability matrix $\edgeP$ for the planted bisection model $G=(n, p, q)$, $p > q$, is  
	such that (1) $\edgeP_{i,j} = p$ for $i, j\in [1, n/2]$ or $i, j \in (n/2, n]$; (2) $\edgeP_{i,j} = q$ otherwise. 
	See Appendix \ref{app:erdos_app} and \ref{app:bisection_app} for the proofs of these results.
	\begin{corollary} \label{coro:1}
		For a Erd\H{o}s-R\'enyi random graph $G = (n, p)$, where $p = \omega\left(\sqrt{\frac{\log n}{n}}\right)$, w.h.p.,  the optimal HC-tree has ratio-cost $\RC^*_G = (1 + o(1)) \frac{2}{3p - p^2}$.
	\end{corollary}
	
	\begin{corollary} \label{coro:2}
		For a planted bisection model $G = (n, p, q)$, where $p > q = \omega\left(\sqrt{\frac{\log n}{n}}\right)$, w.h.p., the optimal HC-tree has ratio-cost $\RC^*_G = (1 + o(1))\frac{2p + 6q}{3(p+q)^2 - p^3 - 3pq^2}$. 
		
		Note that the $\RC^*_G$ decreases as $p$ increases. When $q < \frac{1}{3} p$, $\RC^*_G$ increases as $q$ increases, otherwise, it decreases as $q$ increases.
	\end{corollary}
	
	Note that while the cost of an optimal tree w.r.t. Dasgupta's cost (and also our total-cost) always concentrates around the cost for the expectation graph, as 
	we mentioned above, the value of the \newcost{} depends on the sampled random graph itself.
	For Erd\H{o}s-R\'enyi random graphs, larger $p$ value indicates tighter connection among nodes, making it more clique-like, andthus  $\rho_G^*$ decreases till $\rho_G^*=1$ when $p=1$. (In contrast, note that the expectation-graph $\expG_p$ is a complete graph where all edges have weight $p$; thus it always has $\rho_{\expG_p}^* =1$ no matter what $p$ value it is.) 
	
	For the planted bisection model, increasing $p$ value also strengthens in-cluster connections and thus $\rho_G^*$ decreases. Interestingly, increasing $q$ value when it is small hurts the clustering structure, because it adds more cross-cluster connections thereby making the two clusters formed by vertices with indices from $[1, \frac{n}{2}]$ and from $[\frac{n}{2}+1, n]$ respectively, less evident -- indeed, $\rho_G^*$ increases when $q$ increases for small $q$. 
	However, when $q$ is already relatively large (close to $p$), increasing it more makes the entire graph closer to a clique, and $\rho_G^*$ starts to decreases when $q$ increases for large $q$. Note that such a refined behavior (w.r.t probability $q$) does not hold for the original cost function by Dasgupta.

	\paragraph{Acknowledgements.}
	We would like to thank Sanjoy Dasgupta for very insightful discussions at the early stage of this work, which lead to the observation of the base-cost for a graph. The work was initiated during the semester long program of \emph{``Foundations of Machine Learning"} at Simons Insitute for the Theory of Computer in spring 2017. This work is partially supported by National Science Foundation under grants CCF-1740761, DMS-1547357, and RI-1815697. 
	
	\bibliographystyle{abbrv}
	\bibliography{references}

	\appendix
	
	\section{Missing details from Section \ref{sec:newcostfunction}}
	
	\subsection{Proof for Claim \ref{claim:totalCrelation}}
	\label{appendix:claim:totalCrelation}
	
	We prove the first equality. In particular, fix any pair of leaves $v_i$, $v_j$, and count how many times $w_{ij}$ is added for both sides. On the left-hand side, each time there is a node $k$ such that relation $\myorder{i}{k}{j}$, $\myorder{j}{k}{i}$ or $\{i | j | k\}$ holds, $w_{ij}$ will be added once. There are exactly $|\myleaves(T[\LCA(i, j)])| - 2$ number of such $k$'s, which is the same as how many time $w_{ij}$ will be added to the right-hand side. Hence the first two terms are equal.

	The second equality is straightforward after we plug in the definition of $\Dcost_G(T)$.
	
	\subsection{Proving that \texorpdfstring{$\BC(G) = 0 \Rightarrow \TC_G(T^*) = 0$}{base cost = 0 implies optimal total cost = 0}} \label{app:tc_bc_0}

	Let $G$ denote an input graph with $\BC(G) = 0$. This means that each triplet $\{i, j, k\}$, there can be at most one edge among them, as otherwise, $\BC(G)$ will not be zero. Hence the degree of each node in $G$ will be at most 1. 
	In other words, the collection of edges in $G$ form a matching of nodes in it. 
	Let $a_1, a_2, \cdots, a_n$, and $b_1, b_2, \cdots, b_n$ denote these matched nodes, where $(a_i, b_i) \in E$ and thus these two nodes are paired together. Let $c_1, c_2, \cdots, c_m$ denote the remaining isolated nodes.
	It is easy to verify that the tree $T$ shown in Figure \ref{fig:appendix_tc_bc_0} has $\TC_G(T) = 0$.
	\begin{figure}[ht]
		\centering 
		\includegraphics[width=6cm]{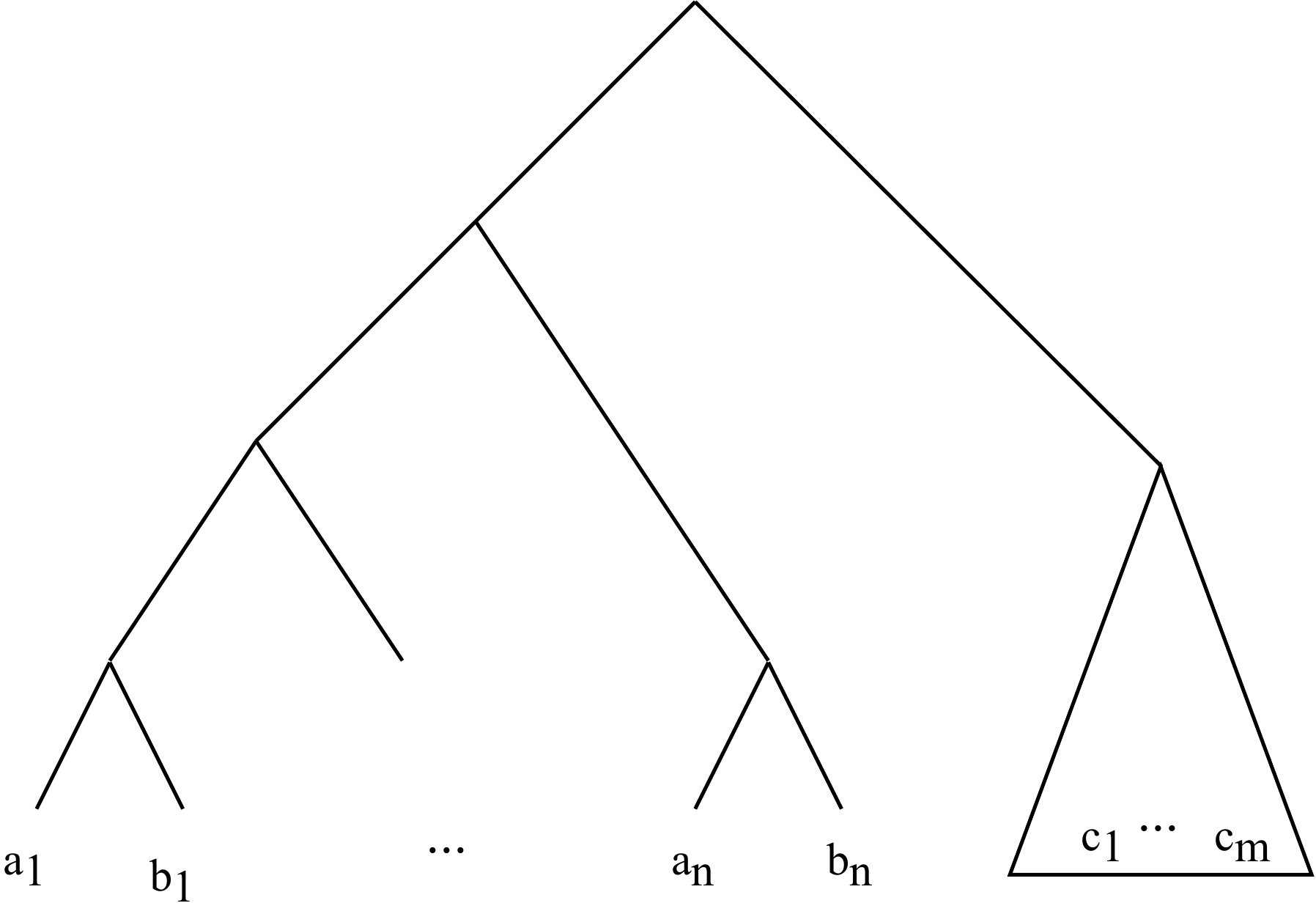}
		\caption{Tree $T$ with $\TC_G(T) = 0$}
		\label{fig:appendix_tc_bc_0}
	\end{figure}

	\subsection{Examples of \newcost{}s of some graphs}
	\label{appendix:examples:costs}
	
	Consider the two examples in Figure \ref{fig:graph_compare} (a). 
	The \emph{linked-stars} $G_1$ consists of two stars centered at $v_1$ and $v_{n/2+1}$ respectively, connected via edge $(v_1, v_{n/2+1})$. 
	For the linked-stars $G_1$, assuming that $n$ is even and by symmetry, we have that:
	\begin{align*}
	\BC(G_1) &= 2\sum_{i\neq j \neq k\in [1,n/2]} \mintricost(i,j,k) + 2\sum_{i\neq j \in [1, n/2], k \in [n/2+1, n]} \mintricost(i,j,k) \\
	&= 2 \sum_{j\neq k \in [2, n/2]} \mintricost(1, j, k) + 2 \sum_{j\in [2, n/2]} \mintricost(1, j, n/2+1) \\
	&= 2\cdot \binom{n/2-1}{2} + 2 (n/2-1) = \frac{n^2}{4} - 2. 
	\end{align*}
	On the other hand, consider the natural tree $T$ as shown in Figure \ref{fig:graph_compare} (b).
	It is easy to verify that $\TC_{G_1}(T)$ equals $\BC(G_1)$ as for any triplet $\triplecost_T = \mintricost$. 
	Hence $\rho_{G_1}^* = 1$, and the linked star has a perfect HC-structure.
	
	For the path graph $G_2$, it is easy to verify that $\BC(G_2) = n-2$. 
	However, as shown in \cite{Dasgupta_2016}, the optimal tree $T^*$ is the complete balanced binary tree and $\Dcost(T^*) = n \log_2 n + O(n)$, meaning that $\TC_{G_2}(T^*) = \Dcost(T^*) - 2|E| = n \log_2 n \pm O(n)$ (via Claim \ref{claim:totalCrelation}). It then follows that 
	$\rho_{G_2}^* = (1+o(1)) \log n$. Intuitively, a path does not possess much of hierarchical structure. 
	
	\begin{figure*}[ht] 
		\centering
		\begin{subfigure}[t]{0.6\textwidth}
			\centering
			\fbox{\includegraphics[height=1.2in]{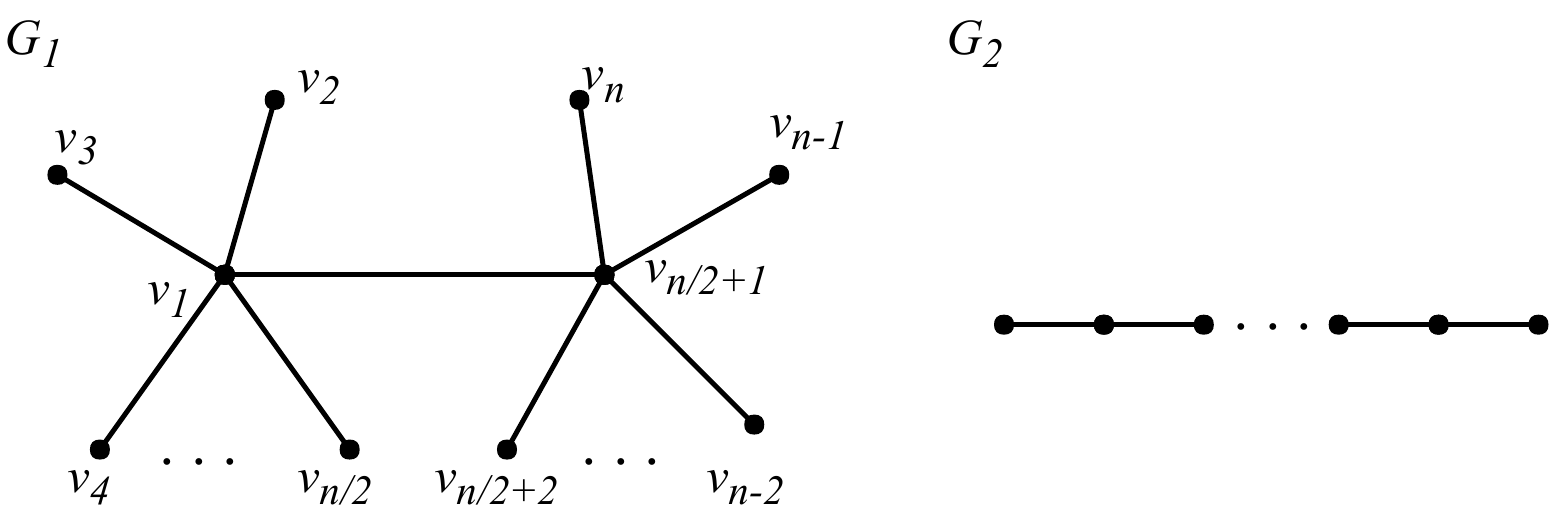}}
			\caption{Linked-star $G_1$ and the path graph $G_2$.}
		\end{subfigure}
		~~~~~ 
		\begin{subfigure}[t]{0.3\textwidth}
			\centering
			\fbox{\includegraphics[height=1.2in]{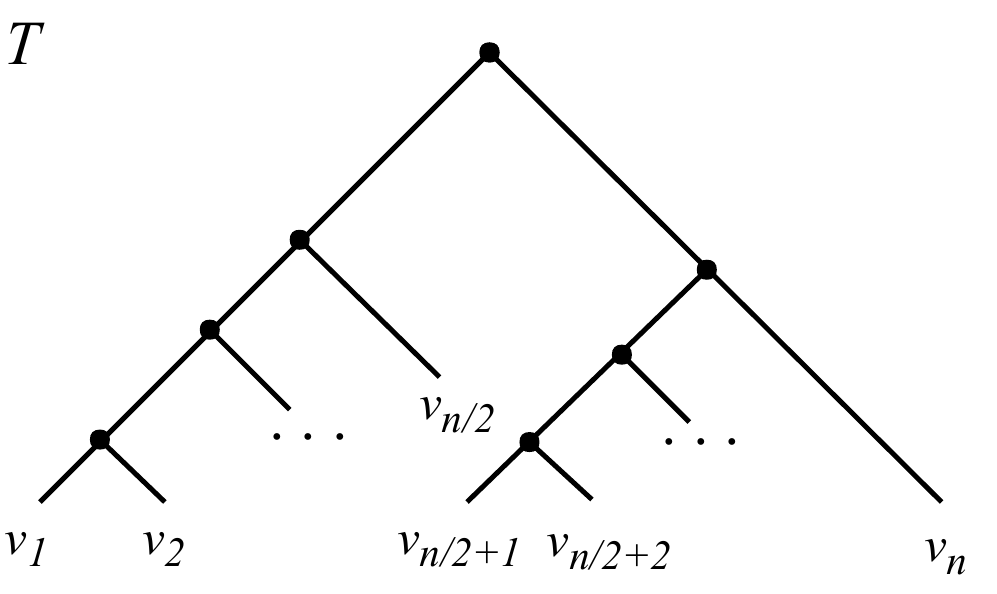}}
			\caption{One optimal tree $T^*$ for linked star $G_1$.
			}
		\end{subfigure}
		\caption{Comparison between a linked star and a path.}
		\label{fig:graph_compare}
	\end{figure*} 
	
	Finally, it is easy to check that, the graph whose similarity (weight) matrix looks like the edge probability matrix for a planted partition as shown in Figure \ref{fig:planted_partition}, also has perfect HC-structure. In particular, as long as $p \ge q$, the tree in Figure \ref{fig:planted_partition} is optimal, and the cost of this tree equals to the base-cost $\BC(G) = 4\cdot p\cdot \binom{n/2}{3} + 2 \cdot (p+q) \cdot n \cdot \binom{n/2}{2}$.  
	More discussions on random graphs are given in Section \ref{sec:random_graphs}. 
	
	\begin{figure}[ht]
		\centering 
		\includegraphics[width=0.55\textwidth]{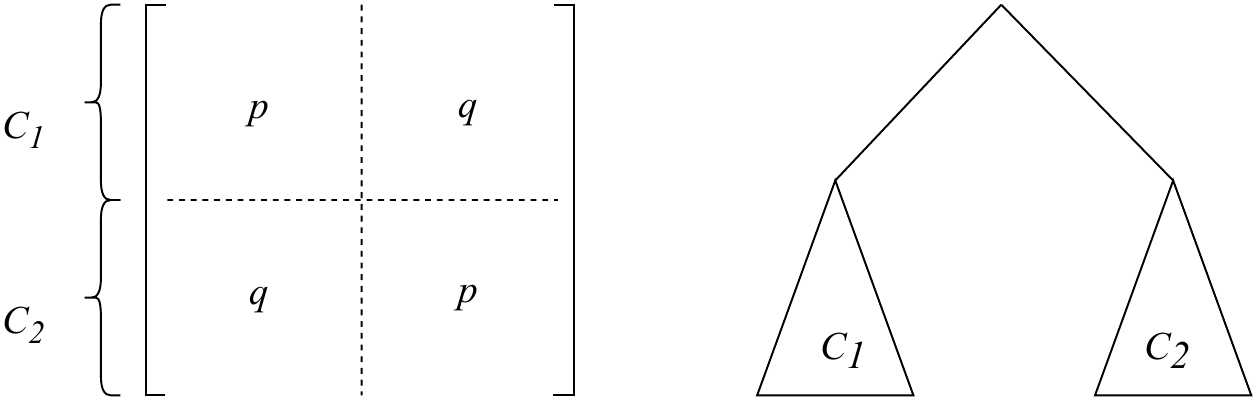}
		\caption{If the weight matrix (on the left) is the same as the probability matrix for a planted partition (where $p > q$), then it has perfect HC-structure. The optimal tree is shown on the right, which first arbitrarily merges all nodes in $C_1$ and $C_2$ individually via any binary tree. }
		\label{fig:planted_partition}
	\end{figure}
	
	\subsection{Proof of Theorem \ref{thm:rhoranges}}
	\label{appendix:thm:rhoranges}

	We prove (i) by induction. 
	Consider the base case where we have a graph $G$ with $n = 3$ and weights $w_{12}, w_{13}$ and $w_{23}$. Assume w.l.o.g that $w_{12}$ is the largest among the three weights. 
	Then $\BC(G) = w_{13}+w_{23}$. 
	Now consider the tree $T$ that merges $v_1$ and $v_2$ first, then merges with $v_3$; for this tree $\TC_G(T) = w_{13} + w_{23}$. Hence $\rho_G^* = 1$ no matter what weights it have. Thus the claim holds for this base case. 
	
	Now assume the claim holds for all $n' \le s$ where $s \ge 3$, and we aim to prove it for a graph $G$ with $n = s+1$. 
	Arrange the indices so that the weight $w_{s, s+1}$ is the largest among all weights. 
	Consider the following HC-tree on nodes of $G$: First, merge $v_s$ and $v_{s+1}$. Next, construct an optimal tree $T'$ for the subgraph $G[V']$ induced by nodes $V' = V \setminus\{s, s+1\} = \{v_1, \ldots, v_{s-1}\}$. Finally, merge  
	$T'$ with the subtree containing $v_s$ and $v_{s+1}$. See figure \ref{fig:bound_induction}. 
	
	\begin{figure}[ht]
		\centering 
		\includegraphics[width=0.37\textwidth]{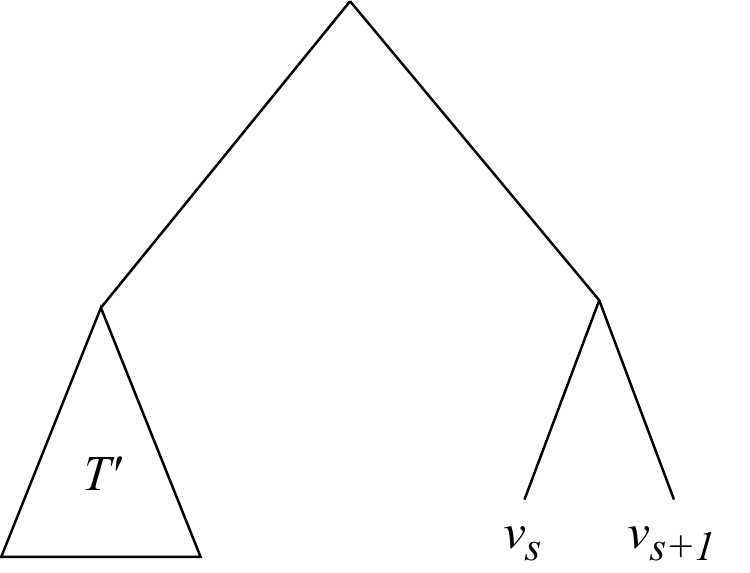}
		\caption{Merge $T'$ with $v_s$ and $v_{s+1}$}
		\label{fig:bound_induction}
	\end{figure}

	Let's for the time being assume that $s-1 \ge 3$ (we will discuss the case when $s-1 = 2$ later). Then by induction on $G[V']$ with $s-1$ nodes, we have: 
	\begin{align}
	\frac{\TC_{G[V']}(T')}{\BC(G[V'])} = \rho_{G[V']}^* \le (s-1)-2 ~~\Rightarrow~~ \TC_{G[V']}(T') \le (s-3) \cdot \BC(G[V']). \label{eqn:inductionhyp}
	\end{align}
	Furthermore, note that
	\begin{align}
	\TC_G(T) &= \TC_{G[V']}(T') + \sum_{i\in[1,s-1]} (w_{is}+w_{i,s+1}) + \sum_{i\neq j \in [1, s-1]} [(w_{is}+w_{i,s+1}) + (w_{js}+w_{j,s+1})] \label{eqn:totalone}
	\end{align}
	For each $i\in[1,s-1]$ (i.e, a node in $T'$), the weight $w_{is}$ will be counted $s-1$ times in the last two terms of the RHS(Eqn. \ref{eqn:totalone}) (where RHS stands for right hand-side); similarly for $w_{i,s+1}$. Combining this with (Eqn. \ref{eqn:inductionhyp}), it then follows that
	\begin{align}
	\TC_G(T) &= \TC_{G[V']}(T') + (s-1)\cdot \sum_{i\in[1,s-1]} (w_{is}+w_{i,s+1}) \\
	&\le (s-3)\cdot \BC(G[V']) + (s-1)\cdot \sum_{i\in[1,s-1]} (w_{is}+w_{i,s+1}). \label{eqn:totaltwo} 
	\end{align}
	\begin{align}
	\BC(G) &= \BC(G[V']) + \sum_{i \in [1, s-1]} (w_{is}+w_{i(s+1)}) + \sum_{i\neq j \in [1, s-1]} [\mintricost(i, j, s) + \mintricost(i, j, s+1)] \nonumber \\
	&\ge \BC(G[V']) + \sum_{i \in [1, s-1]} (w_{is}+w_{i(s+1)}); \label{eqn:baseone} 
	\end{align}
	Combining (Eqn. \ref{eqn:totaltwo}, \ref{eqn:baseone}), we then have that 
	\begin{align}
	\rho_G^* \le \frac{\TC_G(T)}{\BC(G)} \le s-1 = n-2. \label{eqn:rhobound}
	\end{align}
	If $s-1 = 2$, then the first term in both (Eqn. \ref{eqn:totaltwo}) and (Eqn. \ref{eqn:baseone}) vanishes. That is, 
	$$\TC_G(T) = (s-1)\cdot \sum_{i\in[1,s-1]} (w_{is}+w_{i,s+1})~~\text{while}~~ \BC(G) \ge \sum_{i \in [1, s-1]} (w_{is}+w_{i(s+1)}). $$
	Hence the bound in (Eqn. \ref{eqn:rhobound}) still holds. 
	It then follows from induction that $\rho_G^*$ holds for any similarity graph $G$ with $n\ge 3$ nodes; which proves claim (i). 
	
	We now prove claim (ii). 
	First, for node $v_i$, let $d_i$ denote its degree in $G$. 
	For any distinct triplet $i, j, k \in [1,n]$, its induced subgraph in $G$ can have 0, 1, 2, or 3 edges. We call the triplet a \emph{wedge} if its induced subgraph has 2 edges, and a \emph{triangle} if it has 3. 
	It is easy to see that only wedges and triangles will contribute to the base-cost, where a wedge contributes to cost 1, and a triangle contributes to cost 2. 
	\begin{align}
	\BC(G) &= \#wedges + 2 \cdot \#triangles ~\ge~  \frac{2}{3} \sum_i {\binom{d_i}{2}} \nonumber \\
	& =  \frac{2}{3} [\sum_i \frac{d_i^2}{2} - m] 
	~ \ge~  \frac{2}{3} [\frac{2m^2}{n} - m] 
	~ = ~ \frac{2m(2m - n)}{3n}. \label{eqn:baseundirected}
	\end{align}
	To obtain the first inequality in the above derivation, we use the fact that for each node $v_i$, $\binom{d_i}{2}$ counts the total number of wedges having $i$ as the apex, as well as the total number of triangles containing $i$. 
	However, if there is a triangle, it will be counted three times in the summation (while for a wedge will be counted exactly once). 
	Since a triangle will incur a cost of $2$, the first inequality thus follows. 
	The second inequality in (Eqn. \ref{eqn:baseundirected}) essentially follows from the Cauchy-Schwarz inequality and that $\sum_i d_i = 2m$. 
	
	On the other hand, by Claim \ref{claim:totalCrelation}, a trivial bound for $\TC_G(T^*)$, where $T^*$ is an optimal HC-tree, is:
	$$\TC_G(T^*) = \sum_{(i,j)\in E} (|\myleaves(T[\LCA(i, j)]) - 2) \le m \cdot (n-2). $$
	Combining this with (Eqn. \ref{eqn:baseundirected}), we already can obtain a bound that is asymptotically the same as the one in claim (ii). 
	Below we show the bound on the optimal total-cost can be improved to $\TC_G(T^*) \le 2m(n-2)/3$. This leads to the tighter upper bound $\rho_G^* \le \frac{n^2-2n}{2m-n}$ as claimed in (ii). (In particular, we note that for this claimed (tighter) bound, when $m = \binom{n}{2}$, we get $\rho_G^* = 1$ as expected.)

	\paragraph{Proof that $\TC_G(T^*) \le 2m(n-2)/3$ for unweighted graph.} 
	Given $G = (V, E)$, for $\BC_{G}$, Eqn (\ref{eqn:baseundirected}) already shows that 
	\[
	\BC(G) \geq \frac{2m(2m - n)}{3n}.
	\]
	Now, we want to show that $\TC_{G}^* \leq \frac{2m(n - 2)}{3}$. For a fixed tree shape $T$, let $L$ denote the set of leaf nodes in $T$. There are $n!$ one-to-one correspondences (permutations) from $V$ to $L$. Let $\sigma$ denote a certain correspondence, and $\Sigma$ the set of all possible correspondences. Then we take average over all correspondences, where $T_{\sigma}$ denotes the tree $T$ under a certain correspondence $\sigma$,
	\begin{eqnarray*}
		\TC_{G}^* \le \TC_G(T) & \leq & \frac{1}{n!} \sum_{\sigma \in \Sigma}\TC_{G} (T_{\sigma}) 
		=  \frac{1}{n!} \sum_{\sigma \in \Sigma} \sum_{(i, j) \in E} (|\myleaves(T[\LCA(\sigma(i), \sigma(j))])| - 2) \\
		& = & \frac{1}{n!} \sum_{\sigma \in \Sigma} \sum_{l_1, l_2 \in L} (|\myleaves(T[\LCA(l_1,l_2)])| - 2) \mathbf{1}_{(\sigma^{-1}(l_1), \sigma^{-1}(l_2)) \in E} \\
		& = & \frac{1}{n!} \sum_{l_1, l_2 \in L}  (|\myleaves(T[\LCA(l_1,l_2)])| - 2) \cdot \sum_{\sigma \in \Sigma} \mathbf{1}_{(\sigma^{-1}(l_1), \sigma^{-1}(l_2)) \in E} , 
	\end{eqnarray*} 
	where $\sigma^{-1}$ is the inverse of $\sigma$, and it is a one-to-one correspondence from $L$ to $V$. $\mathbf{1}_{(\sigma^{-1}(l_1), \sigma^{-1}(l_2)) \in E}$ is the indicator function, it equals 1 if there is an edge between $\sigma^{-1}(l_1)$ and $\sigma^{-1}(l_2)$. Now, focus on the second summation from the last line of the above equation. The indicator function will have value 1 only when $\sigma^{-1}$ maps $l_1$ and $l_2$ back to two end points of an edge in $E$. There are $m$ edges in $E$ which provide $2m$ possibilities. For other $n - 2$ leaf nodes, $\sigma^{-1}$ can be arbitrary, so there are $(n-2)!$ possibilities. Therefore,  
	\[
	\sum_{\sigma \in \Sigma} \mathbf{1}_{(\sigma^{-1}(l_1), \sigma^{-1}(l_2)) \in E} = 2m (n - 2)! .
	\]
	Plug it back to equation , we have 
	\begin{eqnarray}
		\TC_{G}^* & \leq & \frac{1}{n!} \sum_{l_1, l_2 \in L}  (|\myleaves(T[\LCA(l_1,l_2)])| - 2) \cdot 2m(n-2)! \nonumber \\
		&=& \frac{2m}{n(n-1)} \sum_{l_1, l_2 \in L}  (|\myleaves(T[\LCA(l_1,l_2)])| - 2) 
		\label{eqn:tcopt}
	\end{eqnarray}
	Now observe that $\sum_{l_1, l_2 \in L}  (|\myleaves(T[\LCA(l_1,l_2)])| - 2)$ is simply $\TC_{G_c} (T)$ the total cost of the tree $T$ w.r.t. the complete graph $G_c$ with unit edge-weight. It then follows that 
	$$\sum_{l_1, l_2 \in L}  (|\myleaves(T[\LCA(l_1,l_2)])| - 2) = 2 \binom{n}{3}.$$ 
	Plugging this to Eqn (\ref{eqn:tcopt}), we have that 
	$$\TC_G^* = \TC_G(T^*) \le \TC_G(T) \le \frac{2m}{n(n - 1)} \cdot 2\cdot \binom{n}{3} = \frac{2m(n - 2)}{3}. $$
	
	Putting this together with the bound that $\BC(G) \geq \frac{2m(2m - n)}{3n}$, we thus obtain that $\rho_G^* \le \frac{n^2-2n}{2m-n}$, which finishes the proof of Claim (ii) of Theorem \ref{thm:rhoranges}. 
	
	\subsection{Relation to ground-truth input graph of \cite{Cohen_SODA2018}}
	\label{appendix:groundtruthinput}
	
	We will briefly review the concept of ground-truth input graph of \cite{Cohen_SODA2018} below, and then present missing proofs of our results in the main text. 
	
	Recall that a metric space $(X, d)$ is a \emph{ultrametric} if we have $d(x, y) \leq \max\{d(x, z), d(y, z)\}$  for any $x, y, z\in X$. 
	The ultrametric is a stronger version of metric, and has been used widely in modeling certain hierarchical tree structures (more precisely, the so-called dendograms), such as phylogenetic trees. 
	Intuitively, the authors of \cite{Cohen_SODA2018} consider a graph to be a ground-truth input if it is ``consistent'' with an ultrametrics in the following sense: 
	
	\begin{definition}[Similarity Graphs Generated from Ultrametrics \cite{Cohen_SODA2018}]
		A weighted graph $G = (V, E, w)$ is a \emph{similarity graph generated from an ultrametric}, if there exists an ultrametric $(V, d)$ on the node set $V$ and a non-increasing function $f: \mathbb{R}_{+} \rightarrow \mathbb{R}_{+}$, such that for every $x, y \in V$, $x \not= y$, we have $w(e) = f(d(x, y))$ for $e = (x,y)$ (note, if $(x,y) \notin E$, then $w(e) = 0$). 
	\end{definition}
	
	To see the connection of an ultrametric with a tree structure, they also provide an alternative way of viewing ground-truth input using generating trees. 
	
	\begin{definition}[Generating Tree \cite{Cohen_SODA2018}]
		Let $G = (V, E, w)$ be a weighted graph ($w(e) = 0$ if $e \notin E$). Let $T$ be a rooted binary tree with $|V|$ leaves and $|V| - 1$ internal nodes, where $\mathcal{N}$ and $L$ denote the set of internal nodes and the set of leaves of $T$, respectively. Let $\sigma : L \rightarrow V$ denote a bijection between leaves of $T$ and nodes in $G$. We say that $T$ is a \emph{generating tree for $G$}, if there exists a function $W : \mathcal{N} \rightarrow \mathbb{R}_{+}$, such that for $N_1, N_2 \in \mathcal{N}$, if $N_1$ appears on the path from $N_2$ to root, $W(N_1) \leq W(N_2)$. Moreover, for every $x, y \in V$, $w((x, y)) = W(LCA(\sigma^{-1}(x), \sigma^{-1}(y)))$.
	\end{definition}
	
	\begin{definition}[Ground-truth Input \cite{Cohen_SODA2018}]
		A graph $G$ is a ground-truth input if it is a similarity graph generated from an ultrametric. Equivalently, there exists a tree $T$ that is generating for $G$.
	\end{definition}
	
	\paragraph{Proof of Theorem \ref{thm:ground_truth_is_HC_perfect}.}	
	If $G$ is a ground-truth input, let $T$ be a generating tree for $G$ whose leaf set is exactly $V$. 
	We now show that $\rho_G(T) = 1$, meaning that $T$ must be an optimal HC-tree for $G$ with $\rho^*_G = \rho_G(T)$. 
	
	Indeed, consider any triplet $v_i, v_j, v_k \in V$, and assume w.l.o.g that $N_2 = \LCA(v_i, v_j)$ is a descendant of $N_1 =  \LCA(v_j, v_k)=  \LCA(v_i, v_j, v_k)$. 
	This means that $N_1$ is on the path from $N_2$ to the root of $T$, and by the definition of generating tree, it follows that $w(v_i, v_j) \geq w(v_i, v_k) = w(v_j, v_k)$, and nodes $v_i$ and $v_j$ are first merged in $T$. Therefore, for this triplet, $\triplecost_T(i,j,k) = \mintricost(i,j,k)$. Since this holds for all triplets, we have $\TC_G(T) = \BC(G)$, meaning that $\rho_G(T) = 1$.
	
	On the other hand, consider the graph (with unit edge weight) in Figure \ref{fig:linkage_not_working}, where it is easy to verify that the tree shown on the right satisfies $\rho_G(T) = 1$ and thus is optimal. 
	However, this graph {\bf is not} a ground-truth input as defined in \cite{Cohen_SODA2018}. In particular, in Proposition \ref{prop:groundtruth-unweighted} which we will prove shortly below, we show that a unit-weight graph is a ground-truth input if and only if each connected component of it is a clique. 
	This completes the proof of Theorem \ref{thm:ground_truth_is_HC_perfect}. 
	
	\paragraph{Proof of Proposition \ref{prop:groundtruth-unweighted}.} 
	Assume tree $T$ is the generating tree for our input unweighted graph $G$, which is a ground truth input. Let $C$ and $T_C$ denote a connected component of $G$, and the subtree of $T$ induced by the leaf set $C$, respectively. It is easy to check then $T_C$ is a generating tree for component $C$. Consider any \emph{connected} triplet $\{v_i, v_j, v_k\} \in C$, and assume w.l.o.g that $w_{ij} = w_{ik} = 1$. We now enumerate all possible relations between $v_i$, $v_j$ and $v_k$ in $T_C$; See figure \ref{fig:generating_tree_ijk}. 
	Note that by the definition of a generating tree, for each possible configuration in Figure \ref{fig:generating_tree_ijk}, it is necessary that $w_{jk} = 1$, meaning that there must be an edge between $v_j$ and $v_k$. 
	In other words, if there are two edges $v_jv_i$ and $v_iv_k$ in $C$, then the third edge $v_jv_k$ must also be present in $C$ and thus in $G$. 

	\begin{figure}[ht]
		\centering 
		\includegraphics[width=0.8\textwidth]{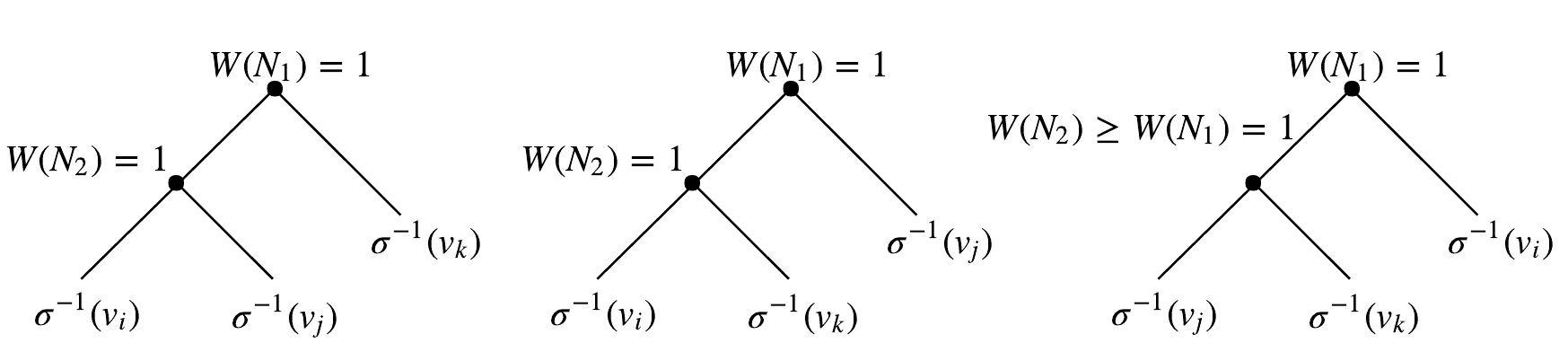}
		\caption{Three possible relations of $v_i$, $v_j$ and $v_k$ in tree $T_C$. In all cases, $w_{jk} = 1$.}
		\label{fig:generating_tree_ijk}
	\end{figure}
	
	Now for any two nodes $v_a$, $v_b$ from the same component $C$, there must be a path, say $v_a = v_{a_0}, v_{a+1}, \cdots, v_{a_k} = v_b$ connecting them. Using the observation above repeatedly, it is easy to see that there must be an edge between any two nodes $v_{a_i}$ and $v_{a_j}$ (including between $v_a$ and $v_b$). This can be made more precisely by an induction on the value of $k$. 
	Since this holds for any two nodes $v_a, v_b$ in the component $C$, $C$ must be a clique. Applying this argument for each component, 
	Proposition \ref{prop:groundtruth-unweighted} then follows.  
	
	\section{Missing details from Section \ref{sec:algorithms}}
	\label{appendix:sec:algorithms}
	
	\subsection{Proof of Theorem \ref{thm:hardness}}
	\label{appendix:thm:hardness}
	
	\paragraph{Proof of Part (1).} 
	Dasgupta shows that it is NP-hard to minimizing $\Dcost_G(T)$ (via converting it to a maximization problem and through a reduction from the so-called \emph{not-all-equal SAT} problem. 
	His reduction is for weighted graphs (with different edge weights). It turns out that a simple modification of his argument shows that minimizing $\Dcost_G(T)$ remains NP-hard even when the input graph has unit edge weight (which we refer to as unweighted graph).  
	We put a sketch of the main modification of Dasgupta's argument in Appendix \ref{app:np_unweighted}. 
	This in turn leads to that minimizing $\rho_G(T)$ is NP-hard even for unweighted graphs. 
	
	\paragraph{Proof of Part (2).}
	Note that if there is a $c$-approximation $\rho_G(T)$ for $\rho_G^* = \rho_G(T^*)$ (i.e, $\rho_G(T) \le c \cdot \rho_G^*$), then, $\TC_G(T)$ is a $c$-approximation for $\TC_G(T^*)$. 
	Since $\Dcost_G(T') = \TC_G(T') + A$ where $A = 2\sum_{i,j \in V} w(i,j)$ for any HC-tree $T'$, it follows that 
	$$\Dcost_G(T) \le c \cdot \TC_G(T^*) + A \le c(\TC_G(T^*) + A) = c \cdot \Dcost_G(T^*). $$
	Hence a $c$-approximation for $\rho_G^*$ gives rise to a $c$-approximation for $\Dcost_G(T^*)$ as well. 
	It then follows that any hardness result in approximating $\Dcost_G(T^*)$ translates into an approximation hardness result for $\rho_G^*$ as well. 
	Claim (2) in the theorem follows directly from the work of \cite{Charikar_SODA2017}, which showed that it is SSE-hard to approximate $\Dcost_G(T^*)$ within any constant factor.

	\subsection{Existence of an \texorpdfstring{$O(\sqrt{\log n})$}{O(sqrt(logn))}-approximation for \texorpdfstring{$\RC_G^*$}{optimal ratio-cost}} \label{app:sdp_app}
	
	It turns out that a slight modification of the SDP algorithm of  \cite{Charikar_SODA2017} gives rise to an $O(\sqrt{\log n})$-approximation algorithm for the ratio-cost $\rho_G^*$ as well. 
	Given that the algorithm is largely the same as the one from \cite{Charikar_SODA2017}, we only sketch the modification here briefly. For details of the original algorithm, see \cite{Charikar_SODA2017} (section 5).
	
	For a fixed graph $G$, the only difference between optimizing $\Dcost_G$ and $\TC_G$  is that we subtract a constant term in the objective function, changing it from 
	\[
	\min \sum_{t=0}^{n-1} \sum_{ij \in E} x_{ij}^t w_{ij}
	\]
	to
	\[
	\min \sum_{t=2}^{n-1} \sum_{ij \in E} x_{ij}^t w_{ij},
	\]
	and keep all other constraints in the formulation of \cite{Charikar_SODA2017}. Because $x_{ij}^t$ will be $1$ for $t = 0, 1$, by doing this, we actually deduct $2 \sum_{ij \in E} w_{ij}$ from the original objective function.
	
	The algorithm of \cite{Charikar_SODA2017} uses a partitioning algorithm from \cite{Krauthgamer_2009} as a subroutine, we will modify slightly. 
	In particular, as our modified objective function does not include terms with $t$ smaller than $2$, we may not have $\sum_{t=|A|/8+1}^{|A|/4} SDP_{A}(t) \leq O(SDP\text{-}HC)$ when the size of cluster $A$ is smaller than $8$.
    Instead, we find the optimal hierarchical tree with brute-force method for small clusters. Let $\TC^*_A$ denotes the optimal total cost over a subset $A$.  
	
	With the similar analysis (\cite{Charikar_SODA2017} section 5.3), the total cost of the tree $T$ produced by the above algorithm is:
	\begin{eqnarray*}
		\TC_{G}(T) &=& \sum_{A, |A| \geq 8} r_{A} \cdot w(E_A) + \sum_{A, |A| < 8} \TC^*_{A} \\
		& \leq & O(\sqrt{\log n}) \sum_{A, |A| \geq 8} \sum_{t=r_A/8+1}^{r_A / 4} SDP_A(t) + \sum_{A, |A| < 8} \TC^*_{A} \\
		& \leq & O(\sqrt{\log n}) SDP\text{-}HC + \TC^*_G \\
		& \leq & O(\sqrt{\log n}) \TC^*_G.
	\end{eqnarray*}

	\subsection{NP-hardness of Minimizing \texorpdfstring{$cost_G$}{cost(G)} on Unweighted Graphs} \label{app:np_unweighted}
	
	The approach is almost the same as the hardness proof by Dasgupta in \cite{Dasgupta_2016} (theorem 10), where he reduces the so-called NAESAT$^*$ problem (a variant of not-all-equal SAT problem) to the problem of \emph{maximizing} the cost (instead of minimizing). In particular, given an instance $\phi$ of the NAESAT$^*$ problem, \cite{Dasgupta_2016} constructs a certain \emph{weighted} graph $G$ of polynomial size as well as a quantity $M$ such that $\phi$ is not-all-equal satisfiable if and only if there exists some tree $T$ so that $\Dcost_G(T) \ge M$. 
	We modify the conversion to obtain an unweighted graph $G$ (i.e, with unit edge weight) still of polynomial size.
	The main idea is that instead of using one single node to represent a literal, we use $r+1$ nodes (itself and $r$ copies, where $r$ will be specified to be $2m$ later). 
    Given the close resemblance of our approach with the argument in \cite{Dasgupta_2016}, we only sketch the modification below. 
    
 In particular, assume that all redundant clauses are already removed in the same way as \cite{Dasgupta_2016}, now we build a graph $G = (V, E)$ with $2n(r+1)$ nodes, $2(r+1)$ per literal ($x_i$, $x_{i1}, \cdots, x_{ir}$, $\bar{x}_i$, $\bar{x}_{i1}, \cdots, \bar{x}_{ir}$). The edge set $E$ falls into three categories, the edges in the first two categories are exactly the same as Dasgupta's construction:
	\begin{enumerate}
		\item For each 3-clause (assume there are $m$ 3-clauses in total), add six edges: three edges joining all 3 literals, and three edges joining their negations.
		\item For each 2-clause (assume there are $m'$ 2-clauses in total), add two edges: one joining two literals, and one joining their negations.
		\item Finally, for each literal and its copies. Make it a complete bipartite graph with $x_i$ and its $r$ copies on one side, $\bar{x_i}$ and its $r$ copies on the other side. See figure \ref{fig:copies}. Only edges in this category connect to copies.
	\end{enumerate}
	Then there will be $6m + 2m' + n(r+1)^2$ edges with unit weight.
	\begin{figure}[ht]
		\centering 
		\includegraphics[width=0.50\textwidth]{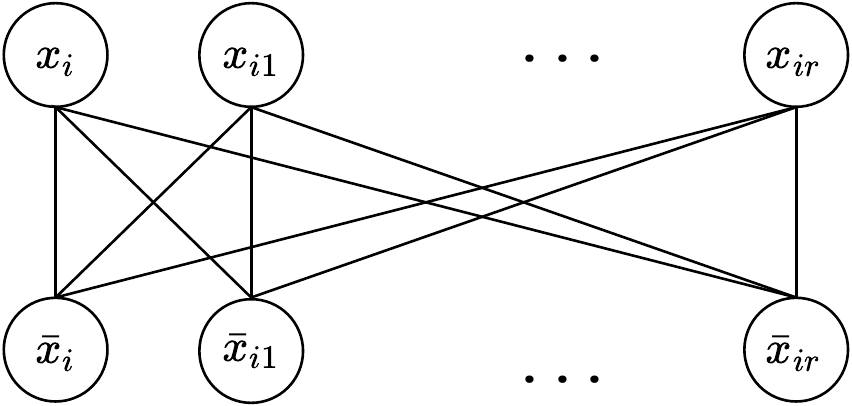}
		\caption{Edges in category 3.}
		\label{fig:copies}
	\end{figure}
	
	Now suppose $\phi$ is not-all-equal satisfiable, and let $V^+$, $V^-$ denote positive and negative nodes under the valid assignment, the copies have the same polarity as its corresponding literal. Similarly as \cite{Dasgupta_2016}, consider the tree $T$ which first separates nodes with different polarities, the cost of the top split is $|V| \cdot |E(V^+, V^-)| = 2n(r+1)(4m + 2m' + n(r+1)^2)$. Then in the second level, all remaining $2m$ edges will be cut since they are disjoint. Therefore, the total cost of tree $T$ will be 
	\begin{eqnarray*}
		cost_G(T) &=& 2n(r+1)(4m + 2m' + n(r+1)^2) + 2n(r+1)m \\
		&=& 2n(r+1)(5m + 2m' + n(r+1)^2).
	\end{eqnarray*}
	Call this $M$. 
	
	Conversely, suppose there is a tree of $cost_G(T) \geq M$, and assume the top split cut it into two parts with size $n(r+1) + \epsilon$ and $n(r+1) - \epsilon$. For this split, it cuts at most $4m$ edges in category 1, $2m'$ edges in category 2, and $(n(r+1) -\epsilon) \cdot (r+1)$ edges in category 3. For the remaining $\epsilon (r+1) + 2m$ edges will add total cost by at most 
	\[
	(n(r+1) + \epsilon) \cdot (\epsilon(r+1) + 2m).
	\] 
	The total cost $\Dcost_G(T)$ is at most
	\begin{eqnarray*}
		\Dcost_G(T) & \leq & 2n(r+1) (4m + 2m' + (n(r+1) - \epsilon) \cdot (r+1)) + (n(r+1) + \epsilon) \cdot (r + 1) \\
		& = & M - 2\epsilon n(r + 1)^2 + 2\epsilon m + \epsilon (r + 1)(\epsilon + n(r+1))
	\end{eqnarray*}
	For $\epsilon \not= 0$, if $cost_G \geq M$, it must be true that
	\[
	2\epsilon m + \epsilon (r + 1) (\epsilon + n(r+1)) \geq 2\epsilon n (r+1)^2.
	\]
	After cancellation, we get
	\[
	2m + \epsilon (r + 1) \geq n (r + 1)^2.
	\]
	We know that $\epsilon \leq n(r+1) - 1$, then we need 
	\[
	2m \geq r + 1,
	\]
	which is impossible if we set $r = 2m$. With $r = 2m$, polynomially in the size of $\phi$, to achieve $cost_G \geq M$, we must first cut it into two equal-sized parts, and necessarily leaves at most one edge per triangle untouched, and also cuts all the other edges, otherwise the cost will fall below $M$ again. Thus, the top split $V \rightarrow (V^+, V^-)$ is a not-all-equal satisfied assignment for $\phi$.
	
	\subsection{Proof of Lemma \ref{lem:withclaw}}
	\label{appendix:lem:withclaw}
	
	First, we state the following simple observation (which follows from the construction of the clusters $\mathcal{C}=\{C_1, \ldots, C_m\}$), which we will use repeatedly later. 
	\begin{claim}\label{claim:typeone}
		If a triplet $\{\av_i, \av_j, \av_k\}$ of $\aG$ is \triTone, then it is not possible that all three nodes are contained in three different clusters in $\mathcal{C}$. 
	\end{claim}
	
	We now claim that all edges $(y_m, y_i)$ for $i \in [4, s]$ must be heavy (recall, $\{y_1, \ldots, y_s\}$ is the maximum clique formed by light edges). 
	This is because otherwise, the triplet $\{ y_1, y_i, y_m\}$ is \triTone{}, which is not possible by Claim \ref{claim:typeone}. 
	Suppose there exists a valid bi-partition $(A, B)$ of $\aV$. 
	Assume that there are two nodes, say $y_i, y_j$ with $i, j\in [1,s]$, are in $A$. 
	The triplet $\{y_i, y_j, y_m\}$ cannot be \triThree{}, as the associated edge weights are not all equal. 
	It also cannot be \triTone{} by Claim \ref{claim:typeone}. Hence $\{y_i, y_j, y_m\}$ is \triTwo{}.  It is then necessary that $y_m \in A$ as well, as $w(y_j, y_m) = w(y_i, y_m) > w(y_i, y_j)$ and the valid bi-partition $(A, B)$ is consistent with this triplet. 
	Furthermore, in this case $B$ can contain at most one vertex from clique $\mathsf{C}$, as otherwise, $y_m$ also has to be in $B$ by using the same argument above, contradicting to that it already must be in $A$. 
	
	Now consider any vertex $y_r \in V'$ outside the clique, and $r\neq m$. 
	Since $y_r \notin \mathsf{C}$, there exists a vertex $y_a \in \mathsf{C}$ such that the edge $(y_r, y_a)$ is heavy. 
	We claim that in this case, all edges $(y_r, y_i)$ with $i\in [1,s]$ must be heavy. 
	Indeed, suppose edge $(y_r, y_b)$ is light for some $b\in [1,s]$, $b\neq a$. 
	Then the triplet $\{y_r, y_a, y_b\}$ is \triTone{}, which is not possible by Claim \ref{claim:typeone}. 
	
	As $y_r$ form heavy edges with all points in the clique $\mathsf{C}$, applying our previous argument showing $y_m\in A$ now to $y_r$, we can prove that $y_r \in A$ as well. 
	In other words, $B\cap V'$ (recall $V' = \{y_1, \ldots, y_m\}$ contains one point from each cluster) is either empty, or it contains exactly one point $y_k$ in which case it is also necessary that $y_k \in \mathsf{C}$, i.e, $k\in [1,s]$. 
	
	Finally, by Proposition \ref{prop:goodcluster}, if a point from a cluster $C_i$ is in $A$ (or $B$), then all points in that cluster are necessarily in $A$ (or $B$). 
	This means that $B\cap V'$ cannot be empty (otherwise $B$ will be empty), and contains exactly one point $y_k$ from the clique $\mathsf{C}$. In this case, the valid partition is the same as $\Pi_k$.  
	This proves Lemma \ref{lem:withclaw}. 
	
	\subsection{Proof of Lemma \ref{lem:noclaw}}
	\label{appendix:lem:noclaw}
	
	Suppose there is a valid bi-partition $(A, B)$, by Proposition \ref{prop:goodcluster} we just need to assign each cluster $C_i$ in $\mathcal{C}$ to either $A$ or $B$. 
	In particular, let $z_1, \ldots, z_m$ be a set of boolean variables with $z_i = 1$ if $C_i \subseteq A$  (thus $z_i = 0$ if $C_i \subseteq B$). 
	By Definition \ref{def:validpartition}, our goal now is to find truth assignments for $z_i$s, so that the resulting bi-partition  $(A, B)$ is consistent with each triplet of $\aG$ of \triTone{} or \triTwo{}. 
	Below we will go through each triplet, and identify constraints on $z_i$s it may incur.

	Consider a triplet $\{\av_i, \av_j, \av_k\}$. 
	If it is \triTone{} with $w_{ij} > \max\{ w_{ik}, w_{jk}\}$, then by Claim \ref{claim:typeone}, it must be that either all three points belong to the same cluster in the valid partition $\mathcal{C}$, or $\av_i$ and $\av_j$ is from the same cluster, while $\av_k$ is from a different one in $\mathcal{C}$. 
	For both cases, it is easy to verify that no matter how $C_i$s are assigned in the bi-partition $(A, B)$, the resulting bi-partition is always consistent with this triplet. 
	
	Now assume the triplet is \triTwo{}.
	If all three nodes are from the same cluster, then it incurs no constraints in the assignments of clusters. 
	Suppose two nodes are from the same cluster, and the third node from a different one. Either this remains the case in the final bi-partition, or all three belong to one subset in the final bi-partition. In either case, the final bi-partition will be consistent with this triplet. 
	
	The only remaining case is when all three nodes in this \triTwo{} triplet are from three different clusters. W.l.o.g, assume this \triTwo{} triplet is $\{\av_i, \av_j, \av_k \}$ from clusters $C_{i}, C_{j}$ and $C_{k}$, respectively. 
	
	Consider an arbitrary point $\av_r \in C_r$ with $r \in [1, m] \setminus \{i, j, k\}$. 
	Suppose for the \triTwo{} triplet $\{\av_i, \av_j, \av_k\}$ is such that $\mathsf{w}=w_{ij}=w_{ik} > w_{jk}$. 
	By Claim \ref{claim:typeone}, for any three points out of $\av_i, \av_j, \av_k$ and $\av_r$, at least two edges have equal weights which is larger than or equal to the last edge weight. 
	Based on this we can show that these 4 points either form a claw (see Figure \ref{fig:claws_and_type2_constraint} (a) and (b)), or $\{\av_r, \av_j, \av_k\}$ also form a \triTwo{} triplet with $w_{rj} = w_{rk} > w_{jk}$ (Figure \ref{fig:claws_and_type2_constraint} (c)). 
	The former case is not possible as we have assumed that there is no claw w.r.t. $\mathcal{C}$. 
	Hence $\av_j$ and $\av_k$ form a \triTwo{} with \emph{every} other $\av_r$, for $r \in [1, m] \setminus \{i, j, k\}$, with $w_{jk}$ having the smallest edge weight in this triplet. If $\av_j$ and $\av_k$ are in the same subset in a valid bi-partition $(A, B)$, then for the bi-partition to be consistent with the triplet $\{\av_r, \av_j, \av_k\}$, $\av_r$ (thus $C_r$) also belong to this subset. 
	In other words, if $z_j = z_k$, then $z_1 = z_2 = \cdots = z_m$ and all clusters belong to the same subset in the bi-partition, which is not possible. 
	It then follows that it must be $z_j \neq z_k$. On the other hand, if $C_j$ and $C_k$ are assigned to different subset in the bi-partition, then no matter how other clusters are assigned, all these triplets $\{\av_r, \av_j, \av_k\}$ will be consistent.
	
	\begin{figure*}[t!] 
		\centering
		\begin{subfigure}[t]{0.3\textwidth}
			\centering
			\includegraphics[height=1.2in]{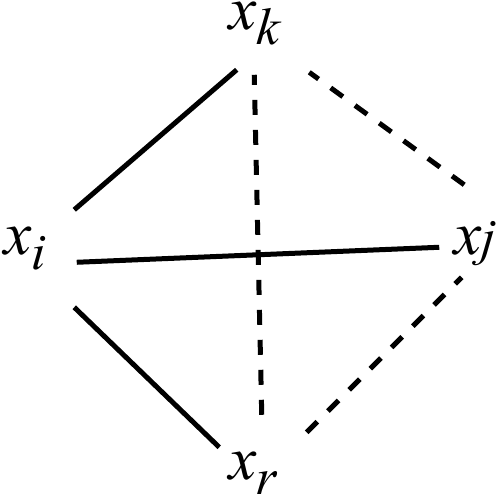}
			\caption{Claw1: $w_{ij} = w_{ik} = w_{ir} > w_{jk}, w_{jr}, w_{kr}$}
		\end{subfigure}
		~ 
		\begin{subfigure}[t]{0.3\textwidth}
			\centering
			\includegraphics[height=1.2in]{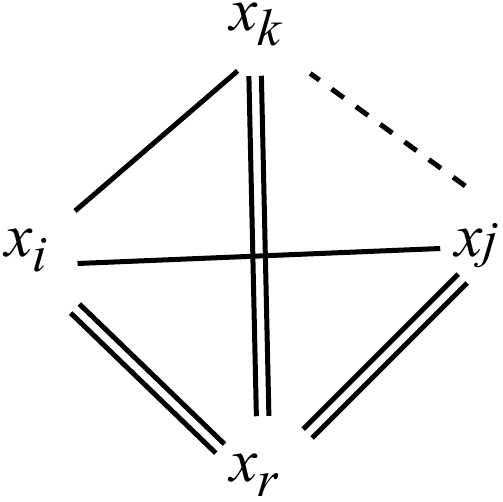}
			\caption{Claw2: $w_{ir} = w_{jr} = w_{kr} > w_{ij} = w_{ik} > w_{jk}$}
		\end{subfigure}
		~ 
		\begin{subfigure}[t]{0.3\textwidth}
			\centering
			\includegraphics[height=1.2in]{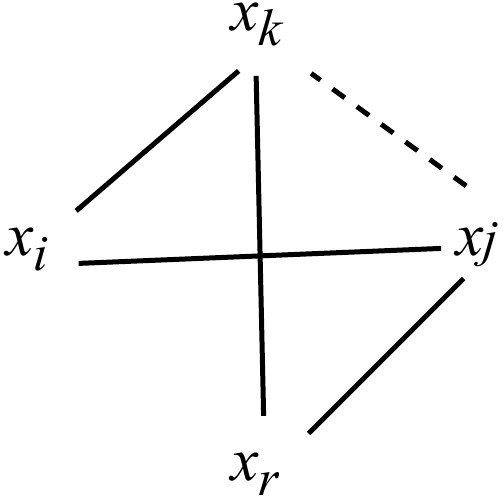}
			\caption{Two \triTwo{} triplets: $w_{ij} = w_{ik} > w_{jk}$, and $w_{jr} = w_{kr} > w_{jk}$.
			}
		\end{subfigure}
		\caption{Three possible patterns consist of four points.}
		\label{fig:claws_and_type2_constraint}
	\end{figure*} 
	
	Finally, going through all $O(n^3)$ \triTwo{} triplets where all three points coming from three distinct clusters, we obtain a set of at most $O(m^2)$ constraints of the form $z_j \neq z_k$ that all need to hold. 
	Given a set of such constraints, we can use a modified breadth-first-search procedure to identify in time linear to the number of constraints whether an assignment of boolean variables $\{z_1, \ldots, z_m\}$ so all constraints are satisfied exists or not, and compute one if it exists. 
	The assignment, if exists, then gives rise to the desired valid bi-partition. If it does not exists, algorithm \validBisect($\aG$) returns ``{\sf fail}". 
	
	\subsection{Proof for Theorem \ref{thm:perfectHC}}
	\label{appendix:thm:perfectHC}
	
	We simply run algorithm \algBisection($G$) on the input graph $G=(V,E)$; let $T$ be its output HC-tree. 
	We claim that $G$ has a perfect HC-structure if and  only if $T$ spans all vertices in $V$; in which case, $T$ is also an optimal HC-tree (i.e, $\rho_G(T) =1$). 
	We first argue the correctness of our algorithm, which is relatively simple. The main technical contribution comes from proving that we can implement the algorithm in the stated time complexity.  
	
	\subsubsection{Correctness of the algorithm}
	First, a subgraph $G'(V')$ of $G$ induced by a subset $V' \subseteq V$ of vertices is the subgraph of $G$ spanned by all edges between vertices in $V'$. 
	We also call $G'(V')$ \emph{an induced subgraph}.
	Recall that by definition of $\TC_G$ and $\BC(G)$, a binary tree $T$ satisfies $\rho_G(T) = 1$ if and only if $T$ is consistent with the similarities of any triplet in $V$ (see the discussion above Definition \ref{def:perfectHC}). 
	If the algorithm returns a HC-tree $T$ spans all vertices in $V$, then we claim that the binary tree $T$ is consistent with all triplets of $\aG$ inductively by considering each subtree $T_v$ rooted at $v \in T$, as well as the corresponding subgraph $G_v$ in a bottom-up manner. 
	First, note that since $T$ spans all vertices in $V$, it means that at each stage  procedure \validBisect($\aG$) succeeds, i.e, it computes a valid bi-partition for $\aG$.

	In particular, at the leaf level, this holds trivially. 
	Now consider an internal node $v \in T$, and the subtree $T_v$ is obtained by \algBisection($G_v$), where $G_v$ is the subgraph of $G$ induced by all leaves in $T_v$. 
	Let $T_A$ and $T_B$ be the two child-subtrees of $T_v$, with $G_A$ and $G_B$ their corresponding subgraphs. By induction hypothesis, $G_A$ (resp. $G_B$) has a perfect HC-structure with $T_A$ (resp. $T_B$) being an optimal HC-tree for it. 
	To check that $T_v$ is also optimal for $G_v$, we just need to verify that for any triplet $\{v_i, v_j, v_k\}$ not all from the same subtrees, $T_v$ is consistent with it. 
	Indeed, this follows easily from that leaves of $T_A$ and $T_B$ form a valid partition of $G_v$.

	For the opposite direction, we need to show that if $G$ has a perfect HC-structure, then \algBisection($G$) should succeed. 
	This follows from the simple claim below, combined with Lemma \ref{lem:withclaw} and \ref{lem:noclaw}. 
	\begin{claim}
		If a graph $\aG = (\aV, \aE)$ has a perfect HC-structure, then any of its induced subgraph also has perfect HC-structure. 
	\end{claim}
	\begin{proof}
		Let $\aT^*$ be an optimal HC-tree for $\aG$. 
		Then $\aT^*$ is binary and it is consistent with any triplet of $\aG$. 
		Let $\aG_A$ be an induced subgraph of $\aG$ spanning on vertices $A \subset \aV$. 
		Now removing all leaves from $\aT^*$ corresponding to vertices $\notin A$, and removing any degree-2 nodes in the resulting tree. 
		This gives rise to an induced binary HC-tree $\aT^*_A$. 
		Obviously, this tree is still consistent with all triplets formed by vertices from $A$. 
		Hence it is an optimal HC-tree for $\aG_A$ with $\rho_{\aG_A}(\aT^*_A) = 1$. 
	\end{proof}
	
	\subsubsection{Time complexity}
	\label{appendix:thm:timecomplexity}
	
	The recursive algorithm \algBisection($G$) has a depth at most $O(n)$. 
	At each level (depth), we call the algorithm for a collection of subgraphs which are all disjoint. Hence at each level, the total size (as measured by size of vertex set) of all subproblems is $O(n)$. 
	We now obtain the time complexity for one subproblem \algBisection($\aG$) with $\nhat$ vertices. 
	
    First, our algorithm builds an initial partition $\mathcal{C} = \{ C_1, \ldots, C_m\}$ of $\aG$. Aho et al. \cite{Aho_1981} gave an algorithm to comptue this partition in $O(\nhat^3 \log \nhat)$ time.
	
	Next, we compute a bi-partition $(A, B)$ from the partition $\mathcal{C}$. 
	If there is no claw w.r.t. $\mathcal{C}$, then Lemma \ref{lem:noclaw} states that the bi-partition can be computed in $O(\nhat^3)$ time. 
	What remains is to bound the time complexity for the case where there exists claws. 
	
	Both checking for claws naively, and checking for bi-partition once a claw is given, takes $O(\nhat^4)$ time. However, there is much structure behind  \emph{crossing triplets}, which are triplets with all three vertices from three different clusters. We will leverage that structure to compute a claw and check for a valid bi-partition in $O(\nhat^3 \log \nhat)$ time. (Recall that any claw will work for our algorithm to compute a bi-partition for this case.)
	
	\paragraph{A more efficient algorithm for identifying a claw.} 
	In particular, we perform the following in algorithm \validBisect($\aG$) to check whether there is a claw or not, and compute one if any exits. Recall that the partition $\mathcal{C} = \{C_1, \ldots, C_m\}$ is already computed. We say that an edge is a \emph{crossing edge} if its two endpoints are from different clusters in $\mathcal{C}$. 
	A triplet is \emph{crossing} if all three nodes inside are from different clusters of $\mathcal{C}$. 
	Let $\pi: \aV \to \{C_1, \ldots, C_m \}$ be such that $\pi(\av)$ is the cluster containing node $\av$. 
	
	Let us now fix a crossing-edge $(\av_j, \av_k)$, say, from clusters $C_j$ and $C_k$, respectively. 
	We scan through all nodes $\av \in C_r$ with $r\neq j, k$ and do the following: 
	\begin{itemize}\denselist
		\item If the triplet $\{\av_j, \av_k, \av_r\}$ is \triThree, we mark cluster $C_r$ with label `0'.  
		\item If the triplet $\{\av_j, \av_k, \av_r\}$ is \triTwo, and $(\av_j, \av_k)$ is one of the edge with the maximum edge weight, then we mark cluster $C_r$ with label `1'. 
		\item Otherwise, the triplet $\{\av_j, \av_k, \av_r\}$ must be \triTwo, and $(\av_j, \av_k)$ must be the edge with minimum edge weight. Let $w = w_{jr}  = w_{kr}$. We mark $C_r$ with label `(2, $w$)'. 
	\end{itemize}
	These are the only choices for the crossing triplet $\{x_j, x_k, x_r\}$ as by Claim \ref{claim:typeone} it cannot be of \triTone. 
	Note that a cluster can receive multiple labels. We mark a cluster with a specific label only if that cluster does not yet have that label. That is, we only maintain distinct labels for a cluster $C_r$. 
	For the fixed crossing-edge $(\av_j, \av_k)$, the total number of labels recorded for all clusters is at most $n$ as they come from at most $n$ crossing triplets. 
	Labeling all clusters for a fixed edge $(\av_j, \av_k)$ takes $O(n)$ time. 
	
	Given a claw $\{\av_a \mid \av_b, \av_c, \av_d\}$, we call the edges $(\av_a, \av_b), (\av_a, \av_c)$ and $(\av_a, \av_d)$ with maximum weight the \emph{legs} of this claw; while the three remaining edges the \emph{base-edges} of this claw. 
	
	\begin{lemma}\label{lem:clawlabels}
		Label all clusters as described above w.r.t. crossing edge $(\av_j, \av_k)$. 
		There is a claw with $(\av_j, \av_k)$ being one base-edge if and only if one of the following holds: \\
		(C-1) there are two clusters $C_r$ and $C_s$ with label `0' and `(2, ?)'; \\
		(C-2) there are two clusters $C_r$ and $C_s$ with label `1' and `(2, ?)'; \\
		(C-3) there are two clusters $C_r$ and $C_s$  with label `(2, $w$)' and `(2, $w'$)' with $w \neq w'$. 
	\end{lemma} 
	\begin{proof}
		By simple case analysis, we can verify that for each of the labeling configuration above, a claw will necessarily be formed (again, the key reason is that by Claim \ref{claim:typeone}, no crossing triplet can be of \triTone{}). 
		For example, suppose configuration (C-1) holds, which means that there exists $\av_r \in C_r$ and $\av_s \in C_s$ such that (i)
		triplet $\{\av_r, \av_j, \av_k\}$ is \triThree{} where all edge weights equal to $w_{jk}$; and (ii) 
		triplet $\{\av_s, \av_j, \av_k\}$ is \triTwo{} where $w := w_{sj} = w_{sk} > w_{jk}$. 
		Now consider the edge $(\av_r, \av_s)$: it must have weight $w_{rs} = w$, as the triplet $\{w_r, w_s, w_j\}$ can only be of \triTwo{} due to Claim \ref{claim:typeone}. 
		Hence $\{\av_s \mid \av_j, \av_k, \av_r\}$ form a claw with $(\av_j, \av_k)$ being a base-edge. 
		The other configurations can be handled similarly. 
		
		On the other hand, suppose we have a claw $\{\av_a \mid \av_b, \av_j, \av_k\}$ with $(\av_j, \av_k)$ being a base-edge. Then again by applying Claim \ref{claim:typeone} we can enumerate possible configurations of triplets $\{\av_a, \av_j, \av_k\}$ and $\{\av_b, \av_j, \av_k\}$ and show that it must be one of the three as claimed through simple case analysis. 
	\end{proof}
	
	\paragraph{Time to identify a claw.}
	We now show that, for a fixed crossing edge $(\av_j, \av_k)$, we can check for the existence of configurations in Lemma \ref{lem:clawlabels} in $O(\nhat \log \nhat)$ time. 
	Indeed, it is easy to identify configurations (C-1) and (C-2) in $O(m) = O(\nhat)$ time, by simply maintaining three flags for each clusters (recording whether it receives a `0', `1', or `2' label). 
	To check for configuration (C-3), we scan labels for clusters $C_1, \ldots, C_m$ as follows. 
	As we scan cluster $C_i$, we maintain a heap $H$ for the distanct weights coming from label `(2, $w$)' already encountered. 
	Now suppose we have already scanned clusters $C_1, \ldots, C_{i-1}$ with $i > 1$, and our current heap is $H_i$. We inspect each label of form `(2, $w$)' associated with $C_i$. 
	If $w$ is not already in $H_{i-1}$, and $H_{i-1}$ is not empty, then we discovered a (C-3) configuration. Our algorithm terminates. 
	Otherwise, if $w$ is already in $H$and $H$ contains more than 1 element, then again we discovered  a (C-3) configuration and the algorithm terminates. 
	Note that during the above checking stage, we do not update heap $H_{i-1}$. 
	
	After we finish inspecting each label in $C_i$, we will update $H_i$. Note that there are in fact only two cases due to our termination conditions above: (i) either $H_{i-1}$ is not empty, in which case $H_i = H_{i-1}$; (ii) or $H_{i-1}$ is empty, and we insert each weight $w$ from label `(2, $w$)' associated to $C_i$ into $H_i$. 
	
	In the worst case, there are $\Theta(\nhat)$ distinct weights associated with the first cluster $C_i$ with non-empty set of labels of the form `(2, ?)'. 
	The entire process takes $O(\nhat \log \nhat)$ time. 
	Since we need to scan through all $O(\nhat^2)$ edges, the following lemma follows. 
	\begin{lemma}\label{lem:clawtime}
It takes $O(\nhat^3 \log \nhat)$ to detect whether a claw exists, and compte one if it exists. 	
	\end{lemma}

	\paragraph{Checking for valid bi-partition with a given claw.}
	Finally, suppose we have now identified a claw $\{y_m \mid y_1, y_2, y_3\}$. 
	Recall in our algorithm for {\it \underline{(Case-1)}} in Section \ref{sec:perfect_near_perfect}, 
	we construct a graph $G'$ induced by a set of representative points $V' = \{y_1, \ldots, y_m\}$ one from each cluster $C_i$. 
	We next compute the subgraph $G''$ consisting only of light edges, 
	and the maximum clique  $\mathsf{C} = \{y_1, \ldots, y_s\}$ containing $\{y_1, y_2, y_3\}$ in $G''$. 
	(In particular, let $w$ denote the weight from the ``legs'' of the claw $\{y_m \mid y_1, y_2, y_3\}$, i.e, $w = w(y_m, y_1) = w(y_m, y_2) = w(y_m, y_3)$. Recall that an edge is \emph{light} if its weight is strictly smaller than $w$, and \emph{heavy} otherwise. ) 
	
	Computing maximum clique in general is expensive, however, in our case it turns out that: 
	\begin{claim}\label{claim:cliquecomponent}
		The component in light-edge graph $G''$ containing $\{y_1, y_2, y_3\}$ is a clique. 
	\end{claim}
	Indeed, in the proof of Lemma \ref{lem:withclaw}, we have already shown that for any $y_r \in V'$ outside the maximum clique $\mathsf{C}$, it is necessary that all edges $(y_r, y_i)$ are heavy for $i\in [1,s]$, which implies the above claim. 
	
	Computing the grpah $G'$, $G''$ and the maximum clique thus takes $O(\nhat^3)$ time. 
	
	What remains is to show given the clique $\mathsf{C} = \{y_1, \ldots, y_s\}$, we can check whether there is a valid bi-partition $\Pi_i$, for $i\in [1, s]$, in $O(\nhat^3)$ time. 
	
	To this end, we maintain an array $L$ of size $s = O(m) = O(\nhat)$, where each entry $L[i]$ will be initialized to be $0$. 
	We scan through all crossing triplet $\{\av_i, \av_j, \av_k\}$. It cannot be \triTone. If it is \triThree, we ignore it. Otherwise, suppose it is \triTwo{} with $w_{ij} = w_{ik} > w_{jk}$. 
	Assuming w.l.o.g that these 3 points are from $C_i, C_j$, and $C_k$ respectively. 
	The only bi-partitions that are not consistent with this triplet is $\Pi_i$, when $\av_j$ and $\av_k$ will be first merged in the resulting binary HC-tree, before merging with $i$. 
	Hence we simply set $L[i] = 1$.

	After we process all crossing triplets in $O(\nhat^3)$ time, we linearly scan array $L$.
	We do not need to process non-crossing triplets as it is already shown at the beginning of proof for Lemma \ref{lem:noclaw} (Section \ref{appendix:lem:noclaw}) that non-crossing triplets will remain valid to any bi-partition arised from merging clusters in $\mathcal{C}$. 
	Hence in the end, if there exists any entry $L[r] =0$, it means that bi-partition $\Pi_r$ must be consistent with all triplets, and thus is valid. Otherwise, there is no valid bi-partition possible. 
	The entire process takes $O(\nhat^3)$ time. 
	Combining this with Lemma \ref{lem:clawtime}, we conclude: 
	
	\begin{lemma}
		There is an algorithm to identify a claw if it exists in $O(\nhat^3 \log \nhat)$ time. If a claw exists, then we can check whether a valid bi-partition exists or not, and compute one if it exists, in $O(\nhat^3)$ time. 
		
		Putting everything together, procedure \validBisect($\aG$) takes $O(\nhat^3 \log \nhat)$ time where $\nhat$ is the number of vertices in $\aG$. 
	\end{lemma}
	
	Finally, this means that each depth level during our recursive algorithm takes $O(\sum_{i} n_i^3 \log n_i) = O(n^3 \log n)$ time, where $n_i$'s are the size of each subgraph at this level and $\sum_i n_i = n$. 
	Since there are at most $O(n)$ levels, the total time complexity of algorithm \algBisection($G$) is $O(n^4\log n)$. 
	This finishes the proof of Theorem \ref{thm:perfectHC}.

	\subsection{Proof of Theorem \ref{thm:deltaperfect}}
	\label{appendix:thm:deltaperfect}

	To prove claim (i), note that $\frac{1}{\delta} \, w(u,v) \le w^*(u,v) \le \delta \, w(u,v)$ for any $u,v\in V$. 
	It then follows that $\BC(G) \geq \frac{1}{\delta} \BC(G^*)$; Furthermore, for any HC-tree $T$, $\TC_{G}(T) \leq \delta \, \TC_{G^*}(T)$. Let $T^*$ be the optimal HC-tree for the graph $G^*$ with perfect HC-structure; $\rho_{G^*}(T^*) =1$. We then have
	\begin{align*}
	\RC_{G}^* & \leq  \frac{\TC_{G}(T^*)}{\BC(G)} 
	\leq  \frac{\delta \, \TC_{G^*}(T^*)}{\frac{1}{\delta} \BC(G^*)} 
	=  \delta^2 \cdot \RC^*_{G^*}(T^*) 
	=  \delta^2,
	\end{align*}
	proving claim (i) of Theorem \ref{thm:deltaperfect}.
	
	We now prove claim (ii). 
	To this end, given the weighted graph $G = (V, E, w)$ with $V = \{v_1, \ldots, v_n \}$, we will first set up a collection $\mathcal{R}$ of \emph{constraints} of the form $\{i, j | k\}$ as follows: 
	\begin{quote}
		Take any triplet $i,j, k \in V$, with weights $w_{ij} \ge w_{jk} \ge w_{ik}$ w.l.o.g. 
		We say that edge $(i,j)$ has \emph{approximately-largest weight among $i,j,k$} if we have $w_{ij}  > \delta^2 w_{jk}$ (and thus $w_{ij} > \delta^2 w_{ik}$ as well).  
		
		If $(i,j)$ has approximately-largest weight, then we add constraint $\{i, j | k\}$ to $\mathcal{R}$. 
	\end{quote}
	
	We aim to compute a binary tree $T$ such that $T$ is \emph{consistent} with all constraints in $\mathcal{R}$; that is, for each $\{i, j | k\} \in \mathcal{R}$, $\LCA(i,j)$ is a proper descendant of $\LCA(i,j,k)$ ((i.e, leaves $v_i$ and $v_j$ are merged first, before merging with $v_k$). 
	
	\begin{figure}[htbp]
		\centering
		\includegraphics[height=3cm]{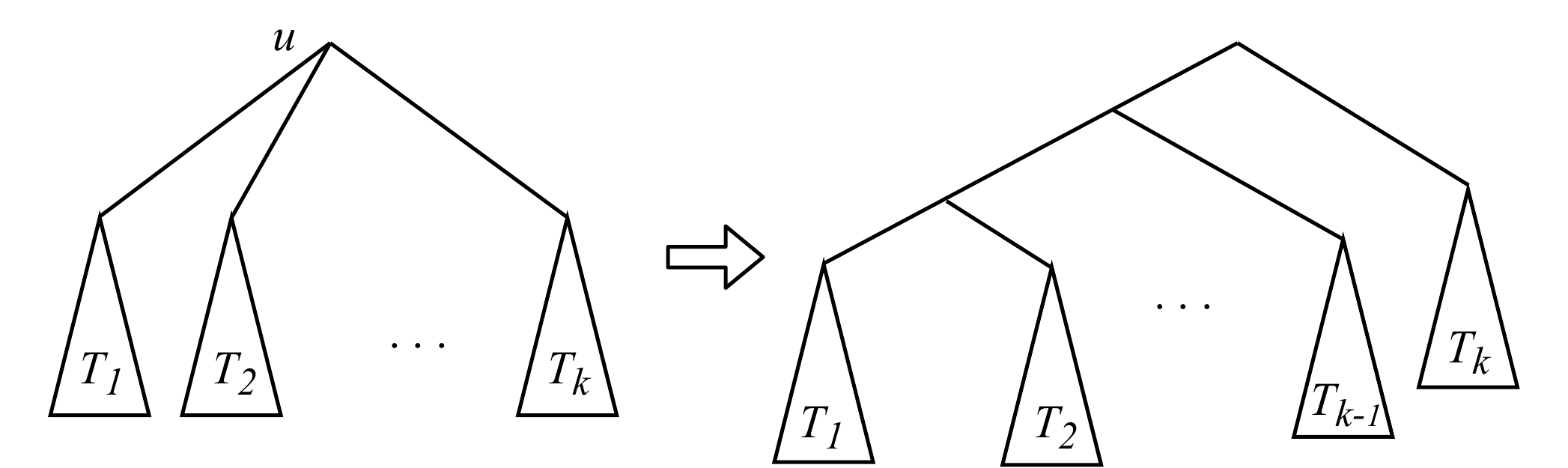}
		\caption{A node $u$ with $k$ children is converted to a collection of binary tree nodes.
			\label{fig:binary}}
	\end{figure} 
	\begin{claim}\label{claim:goodapproxT}
		We can compute a binary HC-tree $T$ consistent with $\mathcal{R}$, if it exists, in $O(n^3)$ time. 
	\end{claim}
	\begin{proof}
		Our problem turns out to be almost the same as the so-called \emph{rooted triplets consistency (RTC)} problem, which has been studied widely in the literature of phylogenetics; e.g, \cite{Byrka_2010_NewResults,Jansson_2005_FastAhoImplementation}. 
		In particular, it is shown in \cite{Jansson_2005_FastAhoImplementation} that a tree $T'$ consistent with a collection of $m$ input constraints on $n$ nodes can be computed in $\min\{O(n + m \log^2 n), O(m + n^2 \log n)\}$ time (if it exists), which is $O(n^3)$ in our case as $m=O(n^3)$. 
		The only difference is that the output tree $T'$ consistent with $\mathcal{R}$ may not be binary. 
		However, consider any node $u\in T'$ with more than 2 children. 
		We claim that we can locally resolve it arbitrarily into a collection binary tree nodes; an example is given in Figure \ref{fig:binary}. 
		It is easy to see that this process does not violate any constraint in $\mathcal{R}$. (We remark that this is due to the fact that $\mathcal{R}$ is generated from only \triTone{} triplets. If there are constraints from \triTwo{} constraints, then this is no longer true. This is why in algorithm \algBisection{} we have to spend much effort to obtain a valid bi-partition in order to derive a binary HC-tree.) 
	\end{proof}
	
	Next, we show that if $G$ has $\delta$-perfect HC-structure, then such a tree $T$ indeed exists. 
	\begin{lemma}\label{lem:approxtree-exists}
		Suppose $G=(V,E,w)$ is a $\delta$-perturbation of a graph $G^* = (V, E, w^*)$ with perfect HC-structure. Then there exists a binary HC-tree $T$ consistent with all constraints in the set $\mathcal{R}$ as constructed above. 
	\end{lemma}
	\begin{proof}
		First, consider an optimal HC-tree $T^*$ for $G^*$. 
		Now given any constraint $\{i, j | k\} \in \mathcal{R}$, note that edge $(i,j)$ must be approximately-largest among $i,j,k$. 
		As $G$ is a $\delta$-perturbation of $G^*$, we have 
		$$w^*_{ij} \ge \frac{1}{\delta} w_{ij} > \frac{1}{\delta} \cdot \delta^2 \max\{ w_{ik}, w_{jk}\} \ge \delta \max\{ \frac{1}{\delta}w^*_{ik}, \frac{1}{\delta}w^*_{jk} \} = \max\{ w^*_{ik}, w^*_{jk} \}, $$
		that is, $w^*_{ij} > \max \{ w^*_{ik}, w^*_{jk} \}$. This is a \triTone{} triplet and thus an optimal tree $T^*$ for $G^*$ has to be consistent with the constraint $\{i, j | k\}$. 
		Hence $T^*$ is consistent with all constraints in $\mathcal{R}$, establishing the lemma. 
	\end{proof}
	 
	\begin{lemma}\label{lem:approxquality}
		Let $T$ be any binary HC-tree that is consistent with constraints in $\mathcal{R}$ as constructed above. Then $\rho_G(T) \le (\delta^2+1) \rho_G^*$. 
	\end{lemma}
	\begin{proof}
		Consider an optimal (binary) HC-tree $\widehat{T}^*$ for the graph $G$. 
		Recall Definition \ref{def:total-cost},
		$$\TC_{G}(\widehat{T}^*)  = \sum_{i,j,k\in V} \triplecost_{\widehat{T}^*, G}(i,j,k). $$
		
		For any triplet $i, j, k \in V$, suppose its order in the optimal HC-tree $\widehat{T}^*$ is $\{i,j|k\}$ (i.e, $i$ and $j$ are first merged, and then merged with $k$ in $\widehat{T}^*$). 
		If the order in $T$ is the same, then $\triplecost_{T, G}(i,j,k) = \triplecost_{\widehat{T}^*, G}(i,j,k)$. 
		Now assume that the order in $T$ is different from the order in $T^*$, say suppose w.l.o.g that relation $\{i,k|j\}$ holds in $T$. Note that $\triplecost_{\widehat{T}^*, G}(i,j,k) = w_{ik}+w_{jk}$. 
		
		There are two possibilities: \\
		(1) $\{i, k| j\}$ is a constraint in $\mathcal{R}$, meaning that $(i,k)$ is approximately-largest among $i,j, k$. 
		Hence 
		\begin{align*}
		\triplecost_{T,G}(i,j,k) = w_{ij}+w_{jk} &\le w_{ik} + w_{jk} \le \triplecost_{\widehat{T}^*, G}(i,j,k). 
		\end{align*}
		
		\noindent (2) $\{i, k| j\}$ is not a constraint in $\mathcal{R}$. 
		Note that as $T$ is consistent with $\{i, k | j\}$, this means that neither $\{j, k | i\}$ nor $\{ i, j | k\}$ can belong to $\mathcal{R}$. So none of the three edges $(i, k)$, $(j, k)$, and $(i, j)$ can be approximately-largest among $i,j,k$. 
		In particular, that $(i,j)$ is \emph{not} approximately-largest among $i,j,k$ means that $w_{ij} \le \delta^2 \max\{ w_{ik}, w_{jk} \}$. If $w_{jk} \le w_{ik}$, then
		\begin{align*}
		\triplecost_{T,G}(i,j,k) = w_{ij}+w_{jk} &\le \delta^2 w_{ik}+w_{jk} \le \delta^2(w_{ik}+w_{jk}) = \delta^2 \cdot \triplecost_{\widehat{T}^*, G}(i,j,k). 
		\end{align*}
		Otherwise $w_{jk} > w_{ik}$, in which case we have: 
		\begin{align*}
		\triplecost_{T,G}(i,j,k) &= w_{ij}+w_{jk} \le \delta^2 w_{jk}+w_{jk}\\
		&\le (1+\delta^2)w_{jk} \le (1+\delta^2)(w_{ik}+w_{jk}) = (1+\delta^2) \cdot \triplecost_{\widehat{T}^*, G}(i,j,k). 
		\end{align*}
		As $\delta \ge 1$, in all cases, we have $\triplecost_{T,G}(i,j,k) \le (1+\delta^2) \cdot \triplecost_{\widehat{T}^*, G} (i,j,k)$ for any $i, j, k \in V$. 
		It then follows that 
		$$\TC_G(T) \le (1+\delta^2) \cdot \TC_G(\widehat{T}^*) ~~\Rightarrow~~ \rho_G(T) \le (1+\delta^2)\rho_G(\widehat{T}^*) = (1+\delta^2)\rho_G^*.
		$$ 
	\end{proof}
	
	Combining Claim \ref{claim:goodapproxT}, Lemma \ref{lem:approxtree-exists} and Lemma \ref{lem:approxquality}, claim (ii) of Theorem \ref{thm:deltaperfect} then follows. 
	
	\section{Missing details from Section \ref{sec:random_graphs}}
	\label{appendix:sec:random_graphs}
	
	\subsection{Proof of Theorem \ref{thm:randomgraph}}
	\label{appendix:thm:randomgraph}
	
    Since our total-cost is closely related to Dasgupta's cost function (recall Claim \ref{claim:totalCrelation}), its proof is almost the same as the argument used in \cite{Cohen_SODA2018} to show the concentration of Dasgupta's cost function on a random graph. Nevertheless, we sketch the proof for completeness. 
	Recall that an optimal tree for a graph w.r.t. \newcost{} function is also optimal w.r.t. the total-cost function. We first claim the following. 
	
	\begin{proposition}\label{prop:total-cost}
		Given a $n \times n$ probability matrix $\edgeP$, assume $\edgeP_{ij} = \omega(\sqrt{\frac{\log n}{n}})$ for $i, j \in [1, n]$. 
		Given a random graph $G=(V,E)$, let $T^*$ be an optimal HC-tree for $G$ w.r.t. \newcost{}, and let $\expT^*$ be an optimal tree for the expectation-graph $\expG$. 
		Then we have that w.h.p., 
		$$|\TC_G(T^*) - \TC_{\expG}(\expT^*)| = o(\TC_{\expG}(\expT^*)). $$ 
	\end{proposition}
	\begin{proof}
		Note that there are $2^{cn\log n}$, for some constant $c>0$, possible HC-trees spanned on $n$ nodes. 
		\begin{claim}\label{claim:singletreeconcentration}
			For an arbitrary but fixed HC-tree $T$, the following holds for some constant $C> c$:
			\[
			\mathbb{P}[~|\TC_{G}(T) - \TC_{\expG}(T)| \geq o(\TC_{\expG}(T))~] \leq \exp (-C n \log n). 
			\]
		\end{claim}
		\begin{proof}
			The proof is almost the same as the proof of Theorem 5.6 of \cite{Cohen_SODA2018} (arXiv version) using a generalized version of Hoeffding's inequality. 
			Specifically, set $Y_{ij} = \mathbf{1}_{(i, j) \in E}$ be the indicator variable of whether $(i,j)$ is in the graph or not. 
			\[
			\text{Set} ~ Z_{ij} = |\myleaves(T[\LCA(i, j)]) - 2| \cdot Y_{ij}; ~~\text{then}~~\TC_{G}(T) = \sum_{i < j} Z_{ij}. 
			\]
			Let $\kappa(n)$ denote the total-cost for a clique of size $n$, which is $2\cdot \binom{n}{3}$. Let $\edgeP_{min}$ denote the smallest entry in $\edgeP$. 
			It is easy to verify that: 
			\begin{enumerate}
				\item $\mathbb{E} [\TC_{G}(T)] = \sum_{i < j} \mathbb{E} [Z_{ij}] =  \sum_{i < j} |\myleaves(T[\LCA(i, j)]) - 2| \cdot \edgeP_{ij} = \TC_{\expG} (T) \geq \kappa(n) \cdot \edgeP_{\min}$.
				\item $ a_{ij} = 0 \leq Z_{ij} \leq |\myleaves(T[\LCA(i, j)]) - 2| = b_{ij}$, for any $i < j \in [1, n]$ \\
				$\sum_{i < j} (b_{ij} - a_{ij})^2 = \sum_{i < j} |\myleaves(T[\LCA(i, j)]) - 2|^2 \leq \binom{n}{2} \cdot n^2$. 
				\item Set $\epsilon = \frac{\sqrt{C \log n / n}}{\edgeP_{\min}} = o(1)$; thus 
				$\epsilon \cdot \TC_{\expG}(T) = o(\TC_{\expG}(T))$.
			\end{enumerate}
			Note that by (1) above, $\mathbb{E} [\TC_{G}(T)] = \TC_{\expG} (T)$. 
			Plug in all these terms into the Hoeffding's inequality, we thus get
			\begin{eqnarray*}
				\mathbb{P}[|\TC_{G}(T) - \TC_{\expG}(T)| \geq \epsilon \cdot \TC_{\expG}(T)] &\leq& \exp \left( -\frac{2 \epsilon^2 \kappa^2(n) \edgeP^2_{\min}}{ \binom{n}{2} \cdot n^2} \right) \\
				&=& \exp \left( - C n \log n \right),
			\end{eqnarray*}
			for a sufficiently large constant $C$, which finishes the proof.
		\end{proof}
		
		As a corollary of the above claim, we have that w.h.p., $| \TC_G(T) - \TC_{\expG}(T) | \le o(\TC_{\expG} (T)$ holds \emph{for all} $O(2^{cn\log n})$ number of HC-trees. 
		It then follows that: 
		\begin{eqnarray*}
			\TC_G(T^*) & \geq & (1 - o(1)) \TC_{\expG}(T^*) 
			\geq  (1 - o(1)) \TC_{\expG}(\expT^*),
		\end{eqnarray*}	
		where the second inequality is due to that $\expT^*$ is an optimal tree for the expectation graph $\expG$ w.r.t. \newcost{}, and thus also w.r.t. the total-cost. Similarly,
		\begin{eqnarray*}
			\TC_G(T^*) & \leq & \TC_{G}(\expT^*) 
			\leq  (1 + o(1)) \TC_{\expG}(\expT^*), 
		\end{eqnarray*}	
		which completes the proof of Proposition \ref{prop:total-cost}.
	\end{proof}

	To bound the denominator $\BC(G)$, we cannot use Hoeffding's inequality directly, because for a triplet, whether they can form a triangle or a wedge are dependent on other triples. We use the following Janson's result from \cite{Janson_2014}.
	
	Let $X = \sum_{\alpha \in \mathcal{A}} Y_{\alpha}$ be a random variable. Let $\Gamma$ be the dependency graph for $\{Y_{\alpha}\}$; that is, each node in $\Gamma$ represents a random variable $Y_{\alpha}$, and two nodes corresponds to $Y_{\alpha}$ and $Y_{\beta}$ will have an edge connecting them if and only if $Y_{\alpha}$ and $Y_{\beta}$ are dependent. Let $\Delta(\Gamma)$ denote the maximum degree among nodes in graph $\Gamma$, and set $\Delta_1(\Gamma) = \Delta(\Gamma) + 1$. 
	
	\begin{theorem}[Theorem 2.3 of \cite{Janson_2014}] \label{thm:janson}
		Suppose that $X = \sum_{\alpha \in \mathcal{A}} Y_{\alpha}$, where $Y_{\alpha}$ is a random variable with $\alpha$ ranging over some index set. If $-c \leq Y_{\alpha} - \mathbb{E} [Y_{\alpha}] \leq b$ for some $b, c > 0$ and all $\alpha \in \mathcal{A}$, and $S = \sum_{\alpha \in \mathcal{A}} \Var Y_{\alpha}$. For $t \geq 0$,
		\[
		\mathbb{P}(X \geq \mathbb{E}[X] + t) \leq \exp \left( - \frac{8t^2}{25 \Delta_1 (\Gamma)(S + bt/3)}\right),
		\]
		and
		\[
		\mathbb{P}(X \leq \mathbb{E}[X] - t) \leq \exp \left( - \frac{8t^2}{25 \Delta_1 (\Gamma)(S + ct/3)}\right).
		\]
	\end{theorem}
	
	\begin{proposition} \label{prop:base-cost}
		Given a $n \times n$ probability matrix $\edgeP$, assume $\edgeP_{ij} = \omega(\frac{\log n}{n})$ for all pair $(i, j)$. For a graph $G$ sampled from $\edgeP$, the following holds w.h.p. 
		\[
		|\BC(G) - \mathbb{E}[\BC(G)]| = o(\mathbb{E}[\BC(G)])
		\] 
	\end{proposition}
	\noindent{\underline{\emph{Remark:}}~} (1) It is important to note that $\mathbb{E}[\BC(G)]$ may be {\it different} from $\BC(\expG)$. (2) Note that $\mathbb{E} [\BC(G)]$ can be calculated easily via linearity of expectation. 
	\begin{proof}
		We first write the random variable $\BC(G)$ as
		\[
		\BC(G) = \sum_{i \not = j \not = k \in [1, n]} Y_{i, j, k}, ~~\text{where}~Y_{i,j,k} := \mintricost(i,j,k). 
		\]
		$Y_{i, j, k}$ is a simple random variable. For simplicity, set $\alpha = (i, j, k)$. Then (1) $Y_\alpha=2$ with probability $\edgeP_{\triangle} = \edgeP_{ij} \cdot \edgeP_{jk} \cdot \edgeP_{ki}$; (2) $Y_\alpha = 1$ with probability $\edgeP_{\wedge} = \edgeP_{ij} \cdot \edgeP_{jk} \cdot (1 - \edgeP_{ki}) + \edgeP_{jk} \cdot \edgeP_{ki} \cdot (1 - \edgeP_{ij}) + \edgeP_{ki} \cdot \edgeP_{ij} \cdot (1 - \edgeP_{jk})$; and (3) otherwise (with probatility $1 - \edgeP_{\triangle} - \edgeP_{\wedge}$), $Y_{\alpha} = 0$. We can easily compute the expectation and variance of $Y_{\alpha}$:
		\[
		\mathbb{E} [Y_{\alpha}] = 2 \edgeP_{\triangle} + \edgeP_{\wedge},
		\]
		\[
		\Var Y_{\alpha} = 4 \edgeP_{\triangle} + \edgeP_{\wedge} - (\mathbb{E} [Y_{\alpha}])^2 < 2 \cdot \mathbb{E} [Y_{\alpha}].
		\]
		Also, for any triplet $\{i, j, k\}$,
		$-2 \leq Y_{i, j, k} - \mathbb{E} [Y_{i, j, k}] \leq 2$,
		and the sum of the variances satisfies: 
		\[
		S = \sum_{i \not = j \not = k \in [1, n]} \Var Y_{i, j, k} < \sum_{i \not = j \not = k \in [1, n]} (2\cdot \mathbb{E} [Y_{i,j,k}]) =  2\cdot \mathbb{E} [\BC(G)].
		\]
		With assumption that $\edgeP_{\min} = \omega(\log n / n)$, for any triplet $\{i, j, k\}$ we have $2\edgeP_{\triangle} + \edgeP_{\wedge} > \edgeP_{ij} \cdot \edgeP_{jk} = \omega(\log^2 n / n^2)$. and
		\begin{align}\label{eqn:upper-expbaseG}
		\mathbb{E} [\BC(G)] &= \sum_{i \not = j \not = k \in [1, n]} \mathbb{E} [Y_{i, j, k}] = \omega(n \log^2 n). 
		\end{align}
		For any two random variable $Y_{\alpha}$ and $Y_{\beta}$, they are dependent only when they share exactly two nodes, i.e., $|\alpha \cap \beta| = 2$. For each triplet, there will be $3n - 9$ dependent triples. Let $\Gamma$ denote the dependency graph,  then $\Delta(\Gamma) = 3n - 9$, and $\Delta_1(\Gamma) = 3n - 8$. 
		
		Set $\epsilon = c\sqrt{\frac{n \log^2 n}{\mathbb{E} [\BC(G)]}}$ for some large constant $c$; note that $\epsilon = o(1)$ due to Eqn (\ref{eqn:upper-expbaseG}). 
		Using Theorem \ref{thm:janson}, we have: 
		\begin{eqnarray*}
			\mathbb{P} ( \BC(G) \geq (1 + \epsilon)\mathbb{E} [\BC(G)]) & \leq & \exp \left( - \frac{8 \epsilon^2 (\mathbb{E}  [\BC(G)])^2}{25(3n - 8)(2\mathbb{E}[\BC(G)] + 2 \epsilon \mathbb{E} [\BC(G)]/ 3)}  \right) \\
			& \leq & \exp \left( -C \frac{\epsilon^2\cdot \mathbb{E}[\BC(G)]}{n}\right) \le n^{-C' \log n}		
		\end{eqnarray*}
		for some constant $C' > 0$. 
		One can establish the other direction in a symmetric manner. Proposition \ref{prop:base-cost} then follows. 
		
	\end{proof}

	We now finish the proof for Theorem \ref{thm:randomgraph}. Combing Propositions \ref{prop:total-cost} and \ref{prop:base-cost}, we have: 
	\begin{eqnarray*}
		\RC^*_G = \RC_G(T^*) &=& \frac{\TC_G(T^*)}{\BC(G)} 
		\le \frac{(1 + o(1)) \TC_{\expG}(\expT^*)}{(1 - o(1)) \mathbb{E}[\BC(G)]} 
		\le (1 + o(1)) \frac{\TC_{\expG}(\expT^*)}{\mathbb{E}[\BC(G)]}.
	\end{eqnarray*}
	
	The lower-bound can be established similarly. Theorem \ref{thm:randomgraph} then follows. 
	
	\subsection{Proof of Corollary \ref{coro:1}} \label{app:erdos_app}
	\begin{proof}
		Following the result of theorem \ref{thm:randomgraph}, we only need to calculate $\TC_{\expG}(\expT^*)$ and $\mathbb{E} [\BC(G)]$, where $\expT^*$ is the optimal tree for the expectation graph $\expG$ w.r.t. the total-cost function (and thus also w.r.t. the \newcost{} function). 
		Since the expectation graph $\expG$ is a complete graph where all edge weights are $p$, 
		we have that
		\[
		\TC_{\expG}(\expT^*) = p \cdot 2\cdot \binom{n}{3}. 
		\]
		On the other hand, let $Y_{i,j,k} = \mintricost(i,j,k)$. $\mathbb{E} [Y_{i,j,k}]= 2 \cdot p^3 + 3\cdot p^2(1-p)$. Hence 
		\[
		\mathbb{E} [\BC(G)] = \binom{n}{3} 2p^3 + \binom{n}{3} 3p^2(1-p).
		\]
		It then follows that 
		\begin{eqnarray*}
			\RC^*_G &=& (1 + o(1)) \frac{2p \binom{n}{3}}{\binom{n}{3} 2p^3 + \binom{n}{3} 3p^2(1-p)} \\
			&=& (1 + o(1))\frac{2p}{ 2p^3 + 3p^2(1-p)} \\
			&=& (1 + o(1))\frac{2}{3p - p^2}
		\end{eqnarray*}
	\end{proof}
	
	\subsection{Proof of Corollary \ref{coro:2}} \label{app:bisection_app}
	\begin{proof}
		Let $\expT^*$ be the  optimal tree for the expectation graph $\expG$ w.r.t. the total-cost function (and thus also w.r.t. the \newcost{} function). Let $V_1 = \{v_1, \ldots, v_{\frac{n}{2}}\}$ and $V_2 = \{v_{\frac{n}{2}+1}, \ldots, v_n\}$. Call $V_1$ and $V_2$ clusters. 
		The expectation-graph $\expG$ is the complete graph where the edge weight between two nodes coming from the same clusters is $p$ while that for a crossing edge is $q$. 
		Hence it is easy to show that the optimal tree $\expT^*$ will first split $V$ into $V_1$ and $V_2$ at the top-most level, and then is an arbitrary HC-tree for each cluster. Hence 
		\[
		\TC_{\expG}(\expT^*) = 2 p \cdot 2 \cdot \binom{n/2}{3} + 4q \cdot {\binom{n/2}{2}} \cdot n/2. 
		\]
		On the other hand, 
		\[
		\mathbb{E} [\BC(G)] = 2 \cdot {\binom{n/2}{3}} \cdot [2p^3 + 3p^2(1-p)] + 2 \cdot {\binom{n/2}{2}} \cdot n/2 \cdot [2pq^2 + q^2(1-p) + 2pq(1-q)].
		\]
		
		It then follows that
		\[
		\RC^*_G = (1 + o(1)) \frac{2p + 6q}{3p^2 + 3q^2 + 6pq - p^3 - 3pq^2}, 
		\]
		which is the main statement of the corollary. 
		
		Now let $E$ denote the fraction in the above bound. 
		We will take partial derivatives to study how changes of $p$ or $q$ will affect the value of $E$. For simplicity, let $f$, $g$ denote the numerator and denominator of the fraction $E$, i.e., $f = 2p + 6q$, and $g = 3p^2 + 3q^2 + 6pq - p^3 - 3pq^2$. First, for $p$, 
		\begin{eqnarray*}
			\frac{\partial E}{ \partial p} &=& \frac{1}{g^2} [2g - f \cdot (6p + 6q - 3p^2 - 3q^2)] \\
			&=& \frac{1}{g^2} (6p^2 + 6q^2 + 12pq - 2p^3 - 6pq^2 - 12p^2 - 12pq + 6p^3 + 6pq^2 -36pq - 36q^2 + 18p^2q + 18q^3) \\
			&=&\frac{1}{g^2}(-6p^2 - 30q^2 - 36pq + 4p^3 + 18q^3 + 18p^2q) \\
			&<& 0.
		\end{eqnarray*}
		The last step follows the fact that $g > 0$, $p, q < 1$, and thus $p^2 > p^3, q^2 > q^3, $ and $pq > p^2$. 
		Hence as $p$ increases, the bound on $\rho_G^*$ decreases. 
		
		Now we consider the partial derivative w.r.t. $q$. 
		\begin{eqnarray*}
			\frac{\partial E}{\partial q} &=& \frac{1}{g^2} [6g - f \cdot (6p + 6q - 6pq)]\\
			&=& \frac{1}{g^2} (18p^2 + 18q^2 + 36pq -6p^3 - 18pq^2 - 12pq - 12p^2 + 12p^2q - 36pq - 36q^2 + 36pq^2) \\
			&=& \frac{1}{g^2} (6p^2 -18q^2 - 12pq - 6p^3 + 12p^2q + 18pq^2).
		\end{eqnarray*}
		Consider $p$ as a fixed constant, and calculate the zero points of parabola $-18(1-p)q^2 - 12p(1-p) q + 6p^2(1-p)$. It is not hard to see that its zeros points are $-p$ and $\frac{p}{3}$, which completes the proof.
	\end{proof}

\end{document}